\documentclass[11pt]{article} 

\usepackage[utf8]{inputenc} 

\usepackage{cite}

\usepackage{hyperref}
\hypersetup{colorlinks=false}

\usepackage{geometry,amsmath} 
\usepackage{graphicx} 
\usepackage{subfigure}
\usepackage{float}

\numberwithin{equation}{section}

\usepackage{booktabs} 
\usepackage{array} 
\usepackage{paralist} 
\usepackage{verbatim} 

\usepackage{todonotes}

\usepackage{amssymb,bm}
\usepackage{amsmath}
\usepackage{amsthm}
\usepackage{mathrsfs}

\theoremstyle{plain}
\newtheorem{thm}{Theorem}[section]

\newtheorem{lem}[thm]{Lemma}
\newtheorem{prop}[thm]{Proposition}

\newtheorem{defn}[thm]{Definition}

\newcommand{\norm}[1]{\left\lVert#1\right\rVert}

\usepackage{fancyhdr} 
\usepackage{sectsty}
\allsectionsfont{\sffamily\mdseries\upshape} 

\title{MIG Median Detectors with  Manifold Filter}

\date{} 
\author{Xiaoqiang Hua{$^{1}$}
~ and Linyu Peng{$^{2}$}\footnote{Corresponding author. E-mail: l.peng@mech.keio.ac.jp.}
\vspace{0.4cm}
\\
{\it 1. College of Meteorology and Oceanography, }\\
{\it National University of Defense Technology, Changsha {410073}, China}\\ 
{\it 2. Department of Mechanical Engineering, Keio University,}\\
{\it  Hiyoshi 3-14-1, Yokohama {223-8522}, Japan}\\
}

\begin{document}
\maketitle
%

\begin{abstract}
In this paper, we propose a class of median-based matrix information geometry (MIG) detectors with a manifold filter and apply them to signal detection in nonhomogeneous environments. As customary, the sample data is assumed to be modeled as Hermitian positive-definite (HPD) matrices, and the geometric median of a set of HPD matrices is interpreted as an estimate of the clutter covariance matrix (CCM). Then, the problem of signal detection can be reformulated as discriminating two points on the manifold of HPD matrices, one of which is the HPD matrix in the cell under test while the other represents the CCM. By manifold filter, we map a set of HPD matrices to another set of HPD matrices by weighting them, that consequently improves the discriminative power by reducing the intra-class distances while increasing the inter-class distances. Three MIG median detectors are designed by resorting to three geometric measures on the matrix manifold, and the corresponding geometric medians are shown to be robust to outliers. Numerical simulations show  the advantage of the proposed MIG median detectors in comparison with their state-of-the-art counterparts as well as the conventional detectors in nonhomogeneous environments.

\end{abstract}




\section{Introduction}

Signal detection in nonhomogeneous environments is an important and fundamental problem in the studies of radar, sonar and communications. Typically, to achieve satisfactory detection performance, the unknown clutter covariance matrix (CCM) should be estimated accurately. The problem of CCM estimation has been extensively studied during the latest decades \cite{GLIGA2019381,8792393,8895800,8747386,7478057,8052571,6166345}. One major limitation of the classical CCM estimators is that the accuracy depends heavily on the  number of homogeneous secondary data, which is independent and identically distributed and is collected from the range cells adjacent to the range cell under test. In practical applications, the number of homogeneous secondary data is often very limited in nonhomogeneous environments. In addition, the secondary data is often contaminated by outliers due to the environmental nonhomogeneity caused by the interfere signal or the variation of clutter power. Both factors make the classical signal detection techniques, such as Kelly's detector \cite{4104190}, adaptive matched filter (AMF) \cite{135446,511809} and adaptive coherence estimator (ACE) \cite{782198} derived by exploiting the CCM estimator suffer from a severe performance degradation.

To address these drawbacks, lots of attempts are made to improve the estimation accuracy of CCM. For instance, De Maio et al. \cite{7426844} exploited a priori information on the covariance structure to design an adaptive detection of point-like target in interference environments. Specifically, a symmetrically structured power spectral density is assumed to agree with the clutter properties. Based on this assumption, three adaptive decision schemes  leveraging on the symmetric power spectral density structure for the interference are proposed, whose performance advantage was confirmed by simulation as well as experimental analysis using  real radar data. In \cite{6166345}, a special structure that the covariance matrix is described as the sum of an unknown positive semi-definite matrix and a matrix proportional to the identity one, and a condition number upper-bound constraint are enforced on the estimated matrix. The constrained structure both accounting for the clutter and the white interference can achieve signal-to-interference-plus-noise ratio (SINR) improvements over their conventional counterparts, especially in the case of limited data. Similar studies concerning the priori information on the structure of the estimated covariance matrix can also be found in \cite{6825699,7605536,6573681,7073532}. Other possible strategies to improve the performance of CCM estimation are mainly focused on resorting to the Bayesian covariance matrix estimation method \cite{4359546,4154721}, using knowledge-based covariance models \cite{5484507,5417154,4267625,5210021,6853408}, or considering shrinkage estimation methods \cite{5484583,6252067}. A fatal drawback on these methods is that the performance gain relies on the known or assumed statistical characteristics of clutter environments, which is often difficult to be acquired in advance in practical applications. 

Recently, a new signal detection scheme designed in the framework of matrix information geometry (MIG) was proposed by Lapuyade-Lahorgue  and Barbaresco \cite{4721049} and developed by Hua et al.  \cite{8000811}, etc. This technique is often referred to as matrix constant false alarm rate detector or MIG detector. MIG studies intrinsic properties of matrix models, and many information processing problems can be transformed into discriminational or optimization problems on matrix spaces equipped with a metric. Based on the MIG theory, an MIG detector designed by using the Riemannian metric or affine invariant Riemannian metric (AIRM) has been successful applied to target detection in coastal X-band and HF surface wave radars \cite{6514112,6450629,4720937}, burg estimation of radar scatter matrix \cite{7842633}, the analysis of statistical characterization \cite{9078971} and the monitoring of wake vortex turbulences \cite{Liu2013,BARBARESCO201054}. The detection framework has been extended by replacing the AIRM distance with other different geometric measures. In \cite{HUA2017106}, the authors developed the MIG detector based on the symmetrized Kullback--Leibler divergence (SKLD) and total Kullback--Leibler divergence. In \cite{HUA2018232,e20040256}, the total Bregman and total Jensen--Bregman divergences were used to design effective MIG detectors. It is noticed  that  MIG detectors designed by using different measures have different detection performances.
As MIG detectors do not rely on the statistical characteristics of clutter environment and take into account the underlying geometry of matrix manifolds, they have the potential to achieve a performance improvement over the conventional detection techniques in nonhomogeneous environments.

 It is well known that a filter can often remove unwanted information from a signal, which can be the noise or the clutter in signal detection. As a consequence, discrimination of the signal and the clutter/noise will be enhanced and hence the signal can be detected more easily. To improve the discriminative power on matrix manifolds, we propose a class of MIG median detectors with a manifold filter, that will improve detection performances by reducing the intra-class distances while increasing the inter-class distances. The main contributions are described as follows:
\begin{enumerate}
  \item We define a map (see Eq. \eqref{eq:manitomani}) on the manifold of Hermitian positive-definite (HPD) matrices by  a manifold filter. Consequently,  each HPD matrix is mapped  to a weighted average of several surrounding  HPD matrices, which is again Hermitian positive-definite.
  \item We reformulate the problem of signal detection into discriminating two points on the HPD matrix manifold, and then design three MIG median  detectors by resorting to the Log-Euclidean metric (LEM), SKLD and the Jensen--Bregman LogDet (JBLD) divergence, respectively. In particular, the median matrices related to the three metrics for a set of HPD matrices are computed by using the fixed-point algorithm and are used as the CCM estimate. Moreover, the anisotropy indices associated with the three metrics are defined to analyze the difference in discriminative power on the HPD manifold and the other manifold.
  \item Simulation results verify that the geometric median matrix is robust and reliable, and the proposed MIG median  detectors with manifold filter can achieve performance improvements over  state-of-the-art counterparts as well as the AMF in nonhomogeneous environments.
\end{enumerate}

The remainder of this paper is organized as follows. Relevant fundamental theories of MIG are introduced in Section \ref{sec:pmig}. Section \ref{sec:pf} is devoted to formulation of  detection problems  on matrix manifolds. Derivation of the detection architecture is presented in Section \ref{sec:det} with numerical simulations given in Section \ref{sec:nr}. Finally, we conclude in Section \ref{sec:con}.

\textit{Notations:} In the sequel, scalars, vectors and matrices are denoted by lowercase, boldface lowercase and boldface uppercase letters, respectively. The symbols $(\cdot)^*$, $(\cdot)^{\operatorname{T}}$ and $(\cdot)^{\operatorname{H}}$ stand for the conjugate, the transpose, and the conjugate transpose, respectively. The operators $\operatorname{det}(\cdot)$ and $\operatorname{tr}(\cdot)$ denote the determinate and trace of a matrix, respectively. The identity $n
\times n$ matrix is denoted by  $\bm{I}_n$ or $\bm{I}$ if no confusion would be caused. We use $\mathbb{C}^n$ and $\mathbb{C}^{n \times n}$ to denote the set of $n$-dimensional complex vectors and $n \times n$ complex matrices, respectively.
%

\section{Preliminary of matrix information geometry}
\label{sec:pmig}

MIG studies intrinsic properties of matrix manifolds, whose elements,  e.g., matrix-valued data, can be discriminated  by a geometric measure. Many problems in information science related to matrix-valued data can be equivalently formulated as measuring the difference of two points in a matrix manifold.  In this section, we briefly introduce the theory of HPD manifolds and define several geometric measures as the discrimination functions using the Frobenius metric.

\subsection{HPD matrix manifolds}

Let $\mathscr{H}(n,\mathbb{C})$ be the set of $n \times n$ Hermitian matrices, given by
\begin{equation}
\mathscr{H}(n,\mathbb{C}) = \left\{ \bm{A} \mid \bm{A} = \bm{A}^{\operatorname{H}}, \bm{A}\in \mathbb{C}^{n\times n} \right\}.
\end{equation}
A Hermitian matrix $\bm{A} \in \mathscr{H}(n,\mathbb{C})$ is positive-definite, expressed as $\bm{A}\succ \bm{0}$, if the quadratic expression $\bm{x}^{\operatorname{H}}\bm{A}\bm{x} > 0$ holds for all $ \bm{x} \in \mathbb{C}^n, \bm{x} \neq {\bf 0}$. The set of $n\times n$ HPD matrices is 
\begin{equation}
\mathscr{P}(n,\mathbb{C}) = \left\{ \bm{A} \mid \bm{A}\succ \bm{0}, \bm{A} \in \mathscr{H}(n,\mathbb{C}) \right\}.
\end{equation}

The subspace $\mathscr{P}(n,\mathbb{C})$ becomes a Riemannian manifold of non-positive curvature equipped with an AIRM \cite{bridson2013metric}, which has been extensively studied during the recent few decades \cite{8624413,8060999,8920217}. The HPD matrix manifold $\mathscr{P}(n,\mathbb{C})$ has many interesting properties \cite{sra2016positive}. For instance,
\begin{enumerate}
  \item Its closure forms a closed, nonpolyhedral, self-dual convex cone.
  \item It embodies a canonical higher-rank symmetric space.
\end{enumerate}

Each element of the matrix manifold $\mathscr{P}(n,\mathbb{C})$ denotes an HPD matrix. Given any two HPD matrices, their difference  on the HPD matrix manifold $\mathscr{P}(n,\mathbb{C})$ can be measured by a metric, a divergence or other measures. Different measures  amount to different geometric structures. Thus, it is important to choose an appropriate measure function  for a specific application. In the following, we introduce several geometric measures on the HPD matrix manifold.

\subsection{Geometric measures on HPD matrix manifolds}



We first introduce the AIRM and AIRM distance on the HPD matrix manifold $\mathscr{P}(n,\mathbb{C})$. The tangent space at a point  $\bm{P}\in \mathscr{P}(n,\mathbb{C})$ is  $T_{\bm{P}}\mathscr{P}(n,\mathbb{C}) = {\bm{P}}\times \mathscr{H}(n,\mathbb{C})$.
A well-known AIRM defined in $T_{\bm{P}}\mathscr{P}(n,\mathbb{C})$ is
\begin{equation}
\langle \bm{A},\bm{B} \rangle_{\bm{P}} := \operatorname{tr}\left(\bm{P}^{-1}\bm{A}\bm{P}^{-1}\bm{B}\right),\quad \bm{A},\bm{B}\in T_{\bm{P}}\mathscr{P}(n,\mathbb{C}). 
\end{equation}


There are natural exponential map and logarithm map on the tangent bundle $T\mathscr{P}(n,\mathbb{C})=\bigcup\limits_{\bm{P}} T_{\bm{P}}\mathscr{P}(n,\mathbb{C})$ using the geodesics. They are related to matrix exponentials and matrix logarithms.   Matrix exponential for a general matrix $\bm{X}$  is given by
\begin{equation}
\exp(\bm{X}) = \sum_{i=0}^\infty \frac{\bm{X}^i}{i!} .
\end{equation}
The following proposition defines the principle logarithm of an invertible matrix.


\begin{lem}[\cite{Hig2008,Moa2005}] \label{lem:aa}
Suppose $\bm{X}$ is an invertible matrix that does not have eigenvalues in the closed negative real line. Then, there exists a unique logarithm with eigenvalues lying in the strip $\{z\in \mathbb{C}\mid -\pi<\operatorname{Im}(z)<\pi\}$. This logarithm is called the principle logarithm and denoted by $\operatorname{Log}\bm{X}$.

Moreover, the matrix $\bm{X}$ satisfies the following properties.
\begin{itemize}
\item[(i)] Both $\bm{X}$ and $\operatorname{Log}\bm{X}$ commute with $[(\bm{X}-\bm{I})s+\bm{I}]^{-1}$ for any real number $s$.
\item[(ii)]   The following identity holds that
\begin{equation*}
\begin{aligned}
\int_0^1[(\bm{X}-\bm{I})s+\bm{I}]^{-2}\operatorname{d}\!s&=(\bm{I}-\bm{X})^{-1}[(\bm{X}-\bm{I})s+\bm{I}]^{-1}\Big|_{s=0}^1\\
&=\bm{X}^{-1}.
\end{aligned}
\end{equation*}
\item[(iii)] Suppose $\bm{A}(\varepsilon)$ is an invertible matrix which does not have eigenvalues lying in the closed real line.
Then, we have
\begin{equation*}
\begin{aligned}
\frac{\operatorname{d}}{\operatorname{d}\!\varepsilon}\operatorname{Log}\bm{A}(\varepsilon)=\int_0^1
[(\bm{A}(\varepsilon)-\bm{I})s+\bm{I}]^{-1}  \frac{\operatorname{d}}{\operatorname{d}\!\varepsilon}\bm{A}(\varepsilon)[(\bm{A}(\varepsilon)-\bm{I})s+\bm{I}]^{-1}\operatorname{d}\!s.
\end{aligned}
\end{equation*}

\end{itemize}
\end{lem}


The AIRM distance of two points $\bm{X}, \bm{Y} \in \mathscr{P}(n,\mathbb{C})$ is defined as  length of the local geodesic connecting them, which is  given by
\begin{equation}
\begin{aligned}
d_R^2(\bm{X},\bm{Y}) &= \norm{ \operatorname{Log}\left(\bm{X}^{-1/2}\bm{Y}\bm{X}^{-1/2}\right) }^2 \\
&= \sum_{i=1}^n \ln^2\lambda_i,
\end{aligned}
\end{equation}
where $\lambda_i$ is the $i$-th eigenvalue of $\bm{X}^{-1/2}\bm{Y}\bm{X}^{-1/2}$. Here $\norm{\cdot}$ is the Frobenius norm $\norm{\bm{A}}^2=\operatorname{tr}(\bm{A}^{\operatorname{H}}\bm{A})$ of a matrix $\bm{A}$, induced from the Frobenius metric.

Unfortunately,  the computational cost of the AIRM distance is expensive in practical applications. An alternative selection, the Log-Euclidean metric (LEM), is derived by defining the metric in the tangent space $T_{\bm{P}}\mathscr{P}(n,\mathbb{C})\ni \bm{A},\bm{B}$ using the Frobenius metric as follows:
\begin{equation}
\langle \bm{A},\bm{B} \rangle_{\bm{P}}^{\operatorname{LE}} := \langle D_{\bm{A}}\operatorname{Log}\bm{P}, D_{\bm{B}}\operatorname{Log}\bm{P} \rangle,
\end{equation}
where $D_{\bm{A}}\operatorname{Log}\bm{P}$ denotes the directional derivative of the logarithm function along direction $\bm{A}$ at a point $\bm{P}$. The LEM metric is more efficient in the sense of computational cost, compared with the AIRM metric. In addition to them, many geometric measures that possess good properties have also been introduced on the manifold $\mathscr{P}(n,\mathbb{C})$, for instance, the JBLD divergence \cite{6378374}, the KLD  \cite{8707051}, and the SKLD  \cite{HUA2017106}. Among these geometric measures, only the AIRM and LEM metrics lead to a true geodesic distance on the manifold $\mathscr{P}(n,\mathbb{C})$. It should also be noted that only the AIRM, the JBLD divergence and the SKLD are invariant w.r.t. affine transformations.
In this paper, we will mainly be focused on the LEM, the JBLD divergence and the SKLD, defined in the manifold $\mathscr{P}(n,\mathbb{C})$.
\begin{defn}
The LEM distance between two matrices $\bm{X}, \bm{Y} \in \mathscr{P}(n,\mathbb{C})$ is defined as
\begin{equation}
d_{L}^2(\bm{X},\bm{Y}) = \norm{ \operatorname{Log}\bm{X} - \operatorname{Log}\bm{Y} }^2.
\end{equation}
\end{defn}

\begin{defn}
The JBLD divergence of two matrices $\bm{X}, \bm{Y} \in \mathscr{P}(n,\mathbb{C})$ is defined as
\begin{equation}
d_J^2(\bm{X},\bm{Y}) = \operatorname{ln}\operatorname{det}\left(\frac{\bm{X}+\bm{Y}}{2}\right) - \frac{1}{2}\operatorname{ln}\operatorname{det}(\bm{X}\bm{Y}).
\end{equation}
\end{defn}

\begin{defn}
The SKLD divergence of two matrices $\bm{X}, \bm{Y} \in \mathscr{P}(n,\mathbb{C})$ is defined as
\begin{equation}
d_S^2(\bm{X},\bm{Y}) = \operatorname{tr}\left( \bm{Y}^{-1}\bm{X} + \bm{X}^{-1}\bm{Y} - 2\bm{I}\right).
\end{equation}
\end{defn}

\section{Problem formulation}
\label{sec:pf}

Consider the detection of a target signal embedding into a clutter environment. We assume that a set of $K$ secondary data, free of signal components, is available. As customary, the detection problem can be formulated as the following binary hypotheses test,
\begin{equation}
\begin{cases}
\mathcal{H}_0: &
\begin{cases}
\bm{x} = \bm{c}  \\
\bm{x}_k = \bm{c}_k , k=1,2,\ldots,K
\end{cases}\vspace{0.2cm}\\
\mathcal{H}_1: &
\begin{cases}
\bm{x} = \bm{c}  \\
\bm{x}_k = \alpha \bm{p} + \bm{c}_k , k=1,2,\ldots,K
\end{cases}\\
\end{cases}
\end{equation}
where $\bm{x}\in\mathbb{C}^n$ and $\bm{x}_k \in \mathbb{C}^n, k=1,2,\ldots,K$ denote the primary data and secondary data, respectively, $\bm{c}\in\mathbb{C}^n$ and $\bm{c}_k\in\mathbb{C}^n, k=1,2,\ldots,K$ are the clutter data, $\alpha\in\mathbb{R}$ is a deterministic but unknown real number, which stands for the target property and channel propagation effects, and  $\bm{p}\in\mathbb{C}^n$ is a known signal steering vector. The purpose of signal detection is to discriminate the target signal from the clutter data. Generally, this problem is considered as the estimation of the logarithm likelihood function of the statistical models of $\mathcal{H}_0$ and $\mathcal{H}_1$, and can be solved by resorting to the generalized likelihood ratio test criteria. One crucial drawback of this method is that it treats the statistical models independently regardless of their intrinsic connections due to the manifold structure. At the moment, we assume that the sample data $\bm{x}$ is modeled as an HPD matrix and then reformulate the problem of signal detection as discriminating the target signal and clutter on the HPD matrix manifold $\mathscr{P}(n,\mathbb{C})$. Specifically, the decision on the presence or absence of a target is made by comparing the geometric distance (or difference) between the matrix in the range cell under test (CUT) $\bm{R}_{CUT}$ and the estimated CCM $\bm{\widehat{R}}$ with a given threshold $\gamma$. Consequently, the rule of signal detection can be formulated as follows:
\begin{equation}
d(\bm{R}_{CUT},\bm{\widehat{R}}) \mathop{\gtrless}\limits_{\mathcal{H}_0}^{\mathcal{H}_1} \gamma,
\label{eq:detect_rule}
\end{equation}
where $d(\cdot, \cdot)$ denotes a geometric measure, and $\gamma$ is the detection threshold that is set to maintain a constant false alarm rate. It can be noted from \eqref{eq:detect_rule} that the detection performance greatly depends on the  measure function used to discriminate two points on the HPD matrix manifold. Moreover,  robustness of the estimated CCM to outliers also has an effect on the detection performance. The commonly used CCM estimator is the sample covariance matrix (SCM). For a set of secondary data $\{ \bm{x}_1,\bm{x}_2, \dots ,\bm{x}_K \}$, the SCM is the maximum likelihood estimation (MLE) of secondary data given by
\begin{equation}
\bm{\widehat{R}}_{SCM} = \frac{1}{K}\sum_{k=1}^K \bm{x}\bm{x}^{\operatorname{H}}.
\label{eq:SCM}
\end{equation}

Note that the SCM estimator can be viewed as the arithmetic mean of $K$ autocorrelation matrices $\{ \bm{x}_k\bm{x}_k^{\operatorname{H}} \}_{k=1}^K$ with rank one. The matrix $\bm{x}_k\bm{x}_k^{\operatorname{H}}$ is the MLE of the secondary data $\bm{x}_k$. To analyze the matrix data on the HPD matrix manifold, for any a sample data $\bm{x}=(x_0,x_1,\ldots,x_{n-1})^{\operatorname{T}}$, a suboptimal estimation that is positive-definite can be derived as
\begin{equation}\label{eq:subopt}
\begin{aligned}
\bm{\widehat{R}} =
\left[ \begin{matrix}
   {r_0} & {r^*_1} & \cdots  & {r^*_{n-1}}  \\
   {r_1} & {r_0} & \cdots  & {r^*_{n-2}}  \\
   \vdots  & \ddots  & \ddots  & \vdots   \\
   {r_{n-1}} & \cdots  & {r_1} & {r_0}  \\
\end{matrix} \right]
\end{aligned}
\end{equation}
with entries
\begin{equation}
r_l = \operatorname{E}[x_ix_{i+l}^*], \quad 0\leq l \leq n-1, 0\leq i \leq n-l-1,
\end{equation}
where $x_l$ is the $l$-th correlation coefficient of the data $\bm{x}$ and $\bold{E}[\cdot]$ denotes the statistical expectation. According to the ergodicity of stationary Gaussian process, the correlation coefficient of the data $\bm{x}$ can be calculated by averaging over time instead of its statistical expectation, as follows:
\begin{equation}
r_l = \frac{1}{n}\sum_{i=0}^{n-l-1} x_ix_{i+l}^*, \quad 0 \leq l \leq n-1.
\end{equation}

The estimation in \eqref{eq:subopt} allow us to model each sample data as an HPD matrix. Then, the CCM can be estimated according to a set of HPD matrices $\{ \bm{\widehat{R}}_1,\bm{\widehat{R}}_2,\ldots,\bm{\widehat{R}}_K \}$. It is known that the SCM is sensitive to outliers in nonhomogeneous environments.  Take the nonlinear geometric structure of HPD matrix manifolds into consideration, a robust CCM estimator can be chosen as the geometric mean or median. The geometric mean estimator has been discussed in details \cite{e20040219}. In this paper, we devote to studying geometric median estimators. The geometric median is defined as follows.
\begin{defn}
For HPD matrices $\{ \bm{\widehat{R}}_1,\bm{\widehat{R}}_2,\ldots,\bm{\widehat{R}}_K \}$, the geometric median estimator associated with a geometric measure is the unique solution of the minimizer of the sum of the geometric measure
\begin{equation}
\bm{\widehat{R}}_g = \underset{\bm{R} \in \mathscr{P}(n,\mathbb{C})}{\operatorname{argmin}} \sum_{k=1}^K d(\bm{R},\bm{\widehat{R}}_k).
\label{eq:def_gm}
\end{equation}
\end{defn}

Based on the geometric median estimator $\bm{\widehat{R}}_g$, we can design the MIG detector by replacing the CCM with $\bm{\widehat{R}}_g$. Different geometric measures possess different discriminative power. Besides, the robustness of different geometric estimators to outliers are also different. These will be clarified  later.

\section{The proposed detectors}
\label{sec:det}

According to the principle of signal detection mentioned above, we present the block-scheme of the proposed detector in Fig. \ref{Block_scheme}. We first model each sample data as an HPD matrix, and then map  each HPD matrix to another HPD matrix. The detection statistic is computed by the geometric distance between the HPD matrix in the CUT and the CCM estimated by the geometric median related to a measure. Finally, the decision on the existence or absence of a target is made by comparing the detection statistic with a given threshold. Here, the manifold-to-manifold map and the geometric median are the two essential tools for designing the detector. In the following, we introduce a manifold-to-manifold map by a manifold filter and deduce geometric medians for a set of HPD matrices. We also define the anisotropy index to analyze the discriminative power of a geometric measure.

\begin{figure}[H]
  \centering
  \includegraphics[width = 8.5cm]{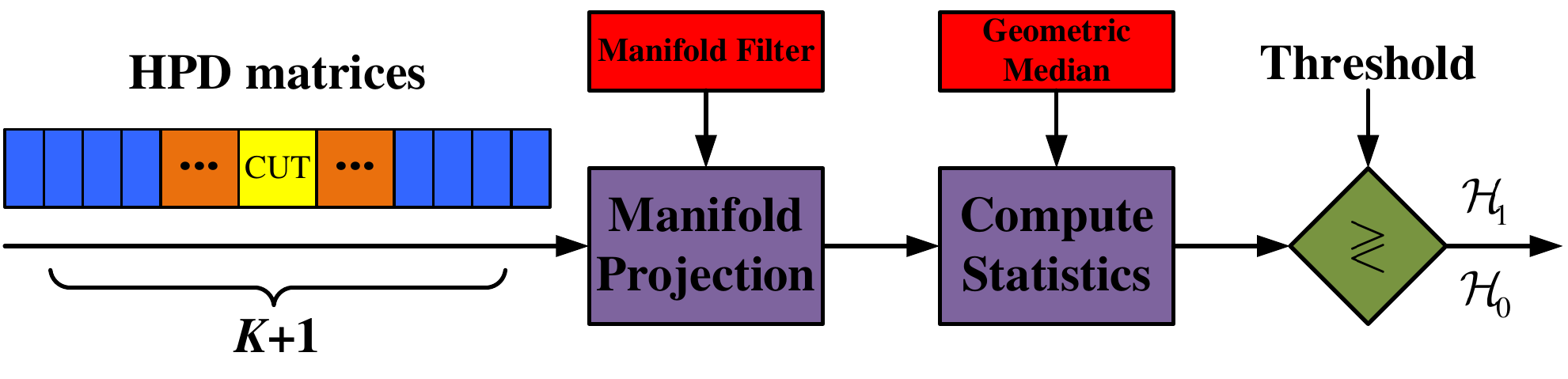}\\
  \caption{Block-scheme of the proposed detector.}
  \label{Block_scheme}
\end{figure}

\subsection{Manifold filter: Weighting a set of HPD matrices}
\label{sec:mantoman}
In this subsection, a weighting map in the HPD manifold is defined to derive a  more discriminative measure of a given set of HPD matrices.

 Suppose that the HPD matrix modeled using the sample data contains some redundant information, which would often limit the discriminative power of the intra-class and the inter-class. As illustrated by Fig. \ref{Fig:Manifold Filter}, given the data collected from two classes on an HPD manifold, marked as the blue circle and the red square, we propose a {\it manifold filter} that cuts down the redundant information contained in the sample data to improve the discriminative power. The improvement is achieved by reducing the intra-class distance while increasing the inter-class distance.

\begin{figure}[H]
  \centering
  \includegraphics[width=8.5cm]{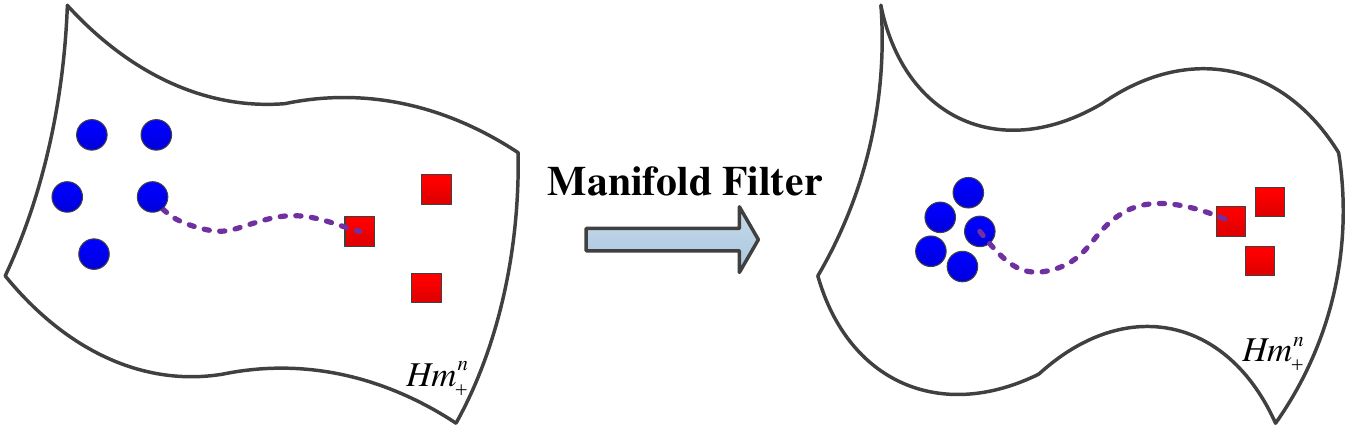}\\
  \caption{Manifold filter.}
  \label{Fig:Manifold Filter}
\end{figure}

The manifold filter is designed by mapping an HPD matrix to the weighted average of its surrounding HPD matrices. Suppose that $\bm{R}$ is an HPD matrix in one cell, and its surrounding $m$ HPD matrices is denoted as $\{\bm{R}_1,\bm{R}_2,\ldots,\bm{R}_m\}$. The weighted average matrix $\bm{\widetilde{R}}$ is given by
\begin{equation}
\bm{\widetilde{R}} := \mathcal{G}\{w_1 \bm{R}_1, w_2 \bm{R}_2, \ldots, w_m \bm{R}_m\},
\label{eq:pre_processing0}
\end{equation}
where 
$w_i$ is the $i$-th weight. The weighted average matrix $\widetilde{\bm{R}}$ can be chosen as the arithmetic mean, geometric mean or geometric median. In each case, the matrix $\widetilde{\bm{R}}$ is still HPD and this defines a map from a HPD matrix to another HPD matrix: $\bm{R}\mapsto\widetilde{\bm{R}}$. Analogous to the filter in image denoising, the arithmetic mean can be viewed as a mean filter \cite{Gavaskar2019Fast} and the geometric mean or median can be viewed as  nonlocal mean filters \cite{8517121}. In this paper, we choose the arithmetic mean as the weighted average matrix. Other means or medians will be reported separately. Thus, the weighted average matrix reads
\begin{equation}
\bm{\widetilde{R}} = \sum_{i=1}^m w_i \bm{R}_i,
\label{eq:pre_processing}
\end{equation}
where  a further constraint of the weights are imposed, namely
\begin{equation}
\sum_{i=1}^m w_i = 1.
\end{equation}
We define each weight $w_i$ to represent the similarity between the matrix $\bm{R}_i$ and the matrix $\bm{R}$ such that the smaller their difference, the greater the weight. For simplicity, we use an exponential function to define the weights as follows
\begin{equation}
w_i = \frac{\exp\left(-d^2(\bm{R}_i,\bm{R})/h^2\right)}{\sum\limits_{j=1}^m \exp\left(-d^2(\bm{R}_j,\bm{R})/h^2\right)},
\label{eq:weight}
\end{equation}
where $h$ is a constant parameter. Substituting \eqref{eq:weight} into \eqref{eq:pre_processing},  we  obtain
\begin{equation}\label{eq:manitomani}
\bm{\widetilde{R}} = \sum_{i=1}^m \frac{\exp\left(-d^2(\bm{R}_i,\bm{R})/h^2\right)}{\sum\limits_{j=1}^m \exp\left(-d^2(\bm{R}_j,\bm{R})/h^2\right)} \bm{R}_i.
\end{equation}

In the manifold filter, $m$ and $h$ are free parameters, which can be adjusted properly in practical applications to  monitor  the filter. In general, bigger $m$ leads to better filter results as more information is included. The parameter $h$ is related to the influence of data on the filter weight: a bigger $h$ usually means a smaller variation of the weights.


\subsection{Geometric medians}

Geometric median of the weighted HPD matrices is considered as the estimate of CCM. According to  \eqref{eq:def_gm}, we can obtain the geometric median related to a measure for a set of HPD matrices. The fixed-point iteration algorithms for computing the LEM, JBLD divergence  and SKLD medians are given below. Be noted that fixed-point iteration for a given fixed-point problem is not unique.

To find the median minimizing the function
\begin{equation}
F(\bm{R})=\sum_{k=1}^K d(\bm{R},\bm{\widetilde{R}}_k),\quad \bm{R} \in \mathscr{P}(n,\mathbb{C}),
\end{equation}
it suffices to solve the vanishment of its gradient, i.e., $\nabla F(\bm{R})=\bm{0}$, which is defined using the Frobenius metric as (see, e.g., \cite{huaetal2020})
\begin{equation}\label{def:gra}
\langle \nabla F(\bm{R}),\bm{X}\rangle:=\frac{\operatorname{d}}{\operatorname{d}\!\varepsilon}\big|_{\varepsilon=0} F(\bm{R}+\varepsilon\bm{X}), \quad \forall\bm{X}\in\mathscr{P}(n,\mathbb{C}).
\end{equation}
In this paper, we show  fixed-point iteration algorithms for solving the algebraic equation $\nabla F(\bm{R})=\bm{0}$. Note that the  algorithms may not be unique.

\begin{prop}
The LEM median of a set of HPD matrices $\{ \bm{\widetilde{R}}_1,\bm{\widetilde{R}}_2,\ldots,\bm{\widetilde{R}}_K \}$ can be computed   via the following fixed-point iteration
\begin{equation}
\begin{aligned}
\bm{R}_{t+1} &= \exp\left\{ \sum_{k=1}^K\frac{\operatorname{Log}\bm{\widetilde{R}}_k}{\norm{ \operatorname{Log}\bm{R}_t - \operatorname{Log}\bm{\widetilde{R}}_k }}\left(\sum_{k=1}^K \frac{1}{\norm{ \operatorname{Log}\bm{R}_t - \operatorname{Log}\bm{\widetilde{R}}_k }}\right)^{-1} \right\},
\end{aligned}
\end{equation}
where $t$ is a nonnegative integer. Uniqueness of the median is proved in  \cite{CCMV2013}.
\end{prop}

\begin{proof}
It was shown in \cite{CCMV2013} that median related to  the function
\begin{equation}
F(\bm{R}) = \sum_{k=1}^K \norm{ \operatorname{Log}\bm{\widetilde{R}}_k -  \operatorname{Log}\bm{R}  }
\end{equation}
is characterized by
%
%
%
\begin{equation}
\sum_{k=1}^K \frac{\operatorname{Log}\bm{\widetilde{R}}_k - \operatorname{Log}\bm{R}}{\norm{\operatorname{Log}\bm{\widetilde{R}}_k - \operatorname{Log}\bm{R}}}=\bm{0}.
\end{equation}
The fixed-point iteration follows immediately. 

%
\end{proof}

\begin{prop}
The JBLD median  of a set of HPD matrices $\{ \bm{\widetilde{R}}_1,\bm{\widetilde{R}}_2,\ldots,\bm{\widetilde{R}}_K \}$ can be computed  by using the fixed-point iteration
\begin{equation}
\begin{aligned}
\bm{R}_{t+1} = \frac{1}{2}\left(\sum_{k=1}^K \frac{1} {\sqrt{\operatorname{ln}\operatorname{det}\left(\frac{\bm{R}_t+\bm{\widetilde{R}}_k}{2}\right) - \frac12 \operatorname{ln}\operatorname{det}\left(\bm{R}_t\bm{\widetilde{R}}_k\right)}} \right)  
\left( \sum_{k=1}^K \frac{\Big(\bm{R}_t+\bm{\widetilde{R}}_k\Big)^{-1}}{\sqrt{\operatorname{ln}\operatorname{det}\left(\frac{\bm{R}_t+\bm{\widetilde{R}}_k}{2}\right) - \frac12 \operatorname{ln}\operatorname{det}\left(\bm{R}_t\bm{\widetilde{R}}_k\right)}} \right)^{-1}.
\end{aligned}
\end{equation}
\end{prop}

\begin{proof}
Denote $F(\bm{R})$ as the function to be minimized, namely
\begin{equation}
F(\bm{R}) = \sum_{k=1}^K \sqrt{\operatorname{ln}\operatorname{det}\left(\frac{\bm{R}+\bm{\widetilde{R}}_k}{2}\right) - \frac{1}{2}\operatorname{ln}\operatorname{det}\left(\bm{R}\bm{\widetilde{R}}_k\right)}.
\end{equation}
Its gradient can be computed via the definition \eqref{def:gra} and we obtain
\begin{equation}
\nabla F(\bm{R})=\frac{1}{4}\sum_{k=1}^K\frac{\left(\frac{\bm{R}+\widetilde{\bm{R}}_k}{2}\right)^{-1}-\bm{R}^{-1}}{\sqrt{\operatorname{ln}\operatorname{det}\Big(\frac{\bm{R}+\bm{\widetilde{R}}_k}{2}\Big) - \frac{1}{2}\operatorname{ln}\operatorname{det}\left(\bm{R}\bm{\widetilde{R}}_k\right)}}.
\end{equation}
Moving the $\bm{R}^{-1}$ terms of the equation $\nabla F(\bm{R})=\bm{0}$ to one side, we obtain the fixed-point iteration.

\end{proof}


\begin{prop}
The SKLD median of a set of HPD matrices $\{ \bm{\widetilde{R}}_1,\bm{\widetilde{R}}_2,\ldots,\bm{\widetilde{R}}_K \}$ can be computed by the fixed-point iteration
\begin{equation}
\begin{aligned}
\bm{R}_{t+1} =\bm{P}_t^{-1/2}\left(\bm{P}_t^{1/2}\bm{Q}_t\bm{P}_t^{1/2}\right)^{1/2}\bm{P}_t^{-1/2},
\end{aligned}
\end{equation}
where
\begin{equation}
\bm{P}_t=\sum_{k=1}^K\frac{\widetilde{\bm{R}}^{-1}_k}{\sqrt{\operatorname{tr}\left(\bm{\widetilde{R}}_k^{-1}\bm{R}_t+\bm{R}_t^{-1}\bm{\widetilde{R}}_k\right)-2n}}
\end{equation}
and
\begin{equation}
\bm{Q}_t=\sum_{k=1}^K\frac{\widetilde{\bm{R}}_k}{\sqrt{\operatorname{tr}\left(\bm{\widetilde{R}}_k^{-1}\bm{R}_t+\bm{R}_t^{-1}\bm{\widetilde{R}}_k\right)-2n}}.
\end{equation}

\end{prop}

\begin{proof}
For the  function
\begin{equation}
F(\bm{R}) = \sum_{k=1}^K \sqrt{\operatorname{tr}\left( \bm{\widetilde{R}}_k^{-1}\bm{R} + \bm{R}^{-1}\bm{\widetilde{R}}_k\right) - 2n},
\end{equation}
direct computation using definition \eqref{def:gra} gives
\begin{equation}
\nabla F(\bm{R})=\frac12\sum_{k=1}^K\frac{\widetilde{\bm{R}}^{-1}_k-\bm{R}^{-1}\widetilde{\bm{R}}_k\bm{R}^{-1}}{\sqrt{\operatorname{tr}\left(\bm{\widetilde{R}}_k^{-1}\bm{R}+\bm{R}^{-1}\bm{\widetilde{R}}_k\right)-2n}}.
\end{equation}
%
The equation $\nabla F(\bm{R})=\bm{0}$ can be rewritten as follows
\begin{equation}\label{eq:Ric}
\bm{R}\bm{P}\bm{R}=\bm{Q},
\end{equation}
which looks like a special algebraic Riccati equation but note that the matrices $\bm{P}$ and $\bm{Q}$ also depend on $\bm{R}$, given by
\begin{equation}
\bm{P}=\sum_{k=1}^K\frac{\widetilde{\bm{R}}^{-1}_k}{\sqrt{\operatorname{tr}\left(\bm{\widetilde{R}}_k^{-1}\bm{R}+\bm{R}^{-1}\bm{\widetilde{R}}_k\right)-2n}}
\end{equation}
and
\begin{equation}
\bm{Q}=\sum_{k=1}^K\frac{\widetilde{\bm{R}}_k}{\sqrt{\operatorname{tr}\left(\bm{\widetilde{R}}_k^{-1}\bm{R}+\bm{R}^{-1}\bm{\widetilde{R}}_k\right)-2n}}.
\end{equation}
Since both $\bm{P}$ and $\bm{Q}$ are HPD matrices, we multiply on both sides of \eqref{eq:Ric} by $\bm{P}^{1/2}$ from both the left and the right, amounting to
\begin{equation}
\bm{P}^{1/2}\bm{R}\bm{P}\bm{R}\bm{P}^{1/2}=\bm{P}^{1/2}\bm{Q}\bm{P}^{1/2}.
\end{equation}
The matrix $\bm{R}$ can hence be solved in terms of $\bm{P}$ and $\bm{Q}$ as follows
\begin{equation}
\bm{R}=\bm{P}^{-1/2}\left(\bm{P}^{1/2}\bm{Q}\bm{P}^{1/2}\right)^{1/2}\bm{P}^{-1/2}.
\end{equation}
The fixed-point iteration is proposed according to the final equation.

\end{proof}

\subsection{Anisotropy index}

The discriminative power of weighted HPD matrices can be shown through the corresponding anisotropy indices. The anisotropy index associated with a measure, e.g., a metric or a divergence, at any point of a  matrix manifold reflects the local geometric structure around this point.

\begin{defn}
The anisotropy index related to a geometric measure at  a point $\bm{R}\in \mathscr{P}(n,\mathbb{C})$ is defined by
\begin{equation}
a(\bm{R}) := \underset{\varepsilon > 0}{\operatorname{min}} \, d^2(\bm{R},\varepsilon\bm{I}).
\label{eq:AI}
\end{equation}
\end{defn}

Namely,  the anisotropy index $a$ is the minimum square of a geometric distance between $\bm{R}$ and $\varepsilon\bm{I}$ ($\varepsilon >0$). This distance represents the projection distance from matrix $\bm{R}$ to the subspace $ \left\{ \varepsilon\bm{I} \mid \varepsilon \in \mathbb{R}^+ \right\}$. Bigger the distance, lagger the anisotropy index. Next, we will show the anisotropy indices for the LEM, JBLD divergence and SKLD, respectively. They can be simply calculated by considering minimum of the function
\begin{equation}
f(\varepsilon)=d^2(\bm{R},\varepsilon\bm{I}), \quad \forall\bm{R}\in\mathscr{P}(n,\mathbb{C}).
\end{equation}

\begin{prop}
The anisotropy index related to the LEM metric at $\bm{R}\in \mathscr{P}(n,\mathbb{C})$ is given by
\begin{equation}
a_L(\bm{R}) = \norm{ \operatorname{Log}\bm{R} - \operatorname{Log}(\varepsilon^*\bm{I}) }^2,
\end{equation}
where  $\varepsilon^* = \sqrt[n]{\operatorname{det}(\bm{R})}$.
\end{prop}

\begin{proof}
Now the function $f(\varepsilon)$ reads
\begin{equation}
f(\varepsilon) = \norm{ \operatorname{Log}\bm{R} - \operatorname{Log}(\varepsilon\bm{I}) }^2,
\end{equation}
whose derivative can immediately be computed:
\begin{equation}
\begin{aligned}
 f'(\varepsilon)  &= -\frac{2}{\varepsilon}\operatorname{tr}\left(\operatorname{Log}\bm{R} - \operatorname{Log}(\varepsilon\bm{I})\right)\\
&=-\frac{2}{\varepsilon}\left( \operatorname{ln}\operatorname{det}(\bm{R}) - \operatorname{ln}\varepsilon^n\right).
\end{aligned}
\end{equation}
Note that Lemma \ref{lem:aa} is used here. Stationary condition that  $f'(\varepsilon)=0$ gives the unique positive solution, denoted by $\varepsilon^*$:
\begin{equation}
\varepsilon^* = \sqrt[n]{\operatorname{det}(\bm{R})}.
\end{equation}
\end{proof}

\begin{prop}\label{prop:aiJ}
The anisotropy index associated with the JBLD divergence at  a point $\bm{R}\in \mathscr{P}(n,\mathbb{C})$ is given by
\begin{equation}
a_J(\bm{R})  = \operatorname{ln}\operatorname{det}\bigg(\frac{\bm{R} + \varepsilon^*\bm{I}}{2}\bigg) - \frac{1}{2}\operatorname{ln}\operatorname{det}(\bm{R})\operatorname{det}(\varepsilon^*\bm{I}),
\end{equation}
where $\varepsilon^*$ is the unique positive solution of the equation
\begin{equation}
2\sum_{i=1}^n \frac{1}{\lambda_i + \varepsilon} = \frac{n}{\varepsilon}
\end{equation}
for all $n\geq 1$, where $\lambda_1,\lambda_2,\ldots,\lambda_n$ are the eigenvalues of $\bm{R}$.
\end{prop}

\begin{proof}
See Appendix \ref{appen:A}.
\end{proof}

\begin{prop}
The anisotropy index associated with the SKLD divergence at  a point $\bm{R}\in \mathscr{P}(n,\mathbb{C})$ is
\begin{equation}
a_S(\bm{R})  = \operatorname{tr}\left(\varepsilon^*\bm{R}^{-1}+\frac{1}{\varepsilon^*}\bm{R} - 2\bm{I}\right),
\end{equation}
where
\begin{equation}
\varepsilon^* = \sqrt{\frac{\operatorname{tr}\left(\bm{R}\right)}{\operatorname{tr}\left(\bm{R}^{-1}\right)}}.
\end{equation}
\end{prop}

\begin{proof}
The proof is similar  by considering  stationary condition of the function
\begin{equation}
f(\varepsilon) = \operatorname{tr}\left(\varepsilon\bm{R}^{-1}+\frac{1}{\varepsilon}\bm{R} - 2\bm{I}\right),
\end{equation}
whose derivative is
\begin{equation}
f'(\varepsilon)=\operatorname{tr}\left(\bm{R}^{-1}-\frac{1}{\varepsilon^2}\bm{R}\right).
\end{equation}
By solving $f'(\varepsilon)=0$, the proof is complete.
%
%
\end{proof}

\section{Numerical results}
\label{sec:nr}

The class of MIG detectors we propose in this paper is based on medians of a set of HPD matrices with a manifold filter,  corresponding to  various  geometric measures. To illustrate the effectiveness of those MIG median detectors, in this section, numerical simulations are provided to analyze them from three aspects:  robustness to outliers;  the discriminative power on the matrix manifold before and after taking the manifold filter, e.g., matrix weighting, into consideration (see Section \ref{sec:mantoman}); and detection performances.

\subsection{Robustness of geometric medians}

In this subsection, we investigate the robustness of geometric medians in terms of the sensitivity to the number of interferences as well as to the number of samples, respectively. The sample dataset is generated according to a multivariate complex Gaussian distribution with zero mean and the following covariance matrix
\begin{equation}
\bm{\Sigma} = \sigma_n^2 \bm{I} + \sigma_c^2\bm{\Sigma}_c,
\end{equation}
where $\sigma_n^2\bm{I}$ and $\sigma_c^2\bm{\Sigma}_c$ denote the noise component with power $\sigma_n^2$ and the clutter component with power $\sigma_c^2$, respectively. The clutter-to-noise ratio (CNR) is defined by CNR$=\sigma_c^2/\sigma_n^2$. Entries of $\bm{\Sigma}_c$ are given by
\begin{equation}
\bm{\Sigma}_c(j,k) = \rho^{|j-k|}\exp\left(\operatorname{i}2\pi f_c(j-k)\right),\quad  j,k = 1,2,\ldots, n,
\label{eq:Known_Cov}
\end{equation}
where $\rho$ is the one-lag correlation coefficient, and $f_c$ denotes the clutter normalized Doppler frequency. In the simulations, the parameter are fixed with $\rho=0.95$, $\sigma_n^2=1$, $\textrm{CNR}=20$ dB, $f_c = 0.1$. A set of $40$ sample data is generated without interferences, and then, a number of $1,2,\ldots,15$ interferences are respectively injected into the sample data. The normalized Doppler frequency of all interferences is set to be $0.2$. The SCNR is defined as
\begin{equation}
\textrm{SCNR} = |\alpha|^2\bm{p}\bm{\Sigma}^{-1}\bm{p},
\end{equation}
where $\alpha$ denotes the amplitude coefficient and $\bm{p}$ is the steering vector. Here, SCNR is set to be $15$ dB. An offset error $L_{error}$ that denotes the difference between the median $\bm{R}_0$ of the sample data without interferences and the median $\bm{R}_{interf}$ of the sample data with interferences is defined by
\begin{equation}
L_{error} = \frac{\norm{ \bm{R}_0 -  \bm{R}_{interf} }}{\norm{ \bm{R}_0 }}.
\end{equation}

Fig. \ref{Fig:Robustness_of_Num_of_Outlier} shows  offset errors of the SCM and geometric medians under different number of interferences.  It is clear that geometric medians are more robust than the SMC, and among them, the JBLD median is the most robust against interferences.

\begin{figure}[H]
  \centering
  \includegraphics[width=8.5cm]{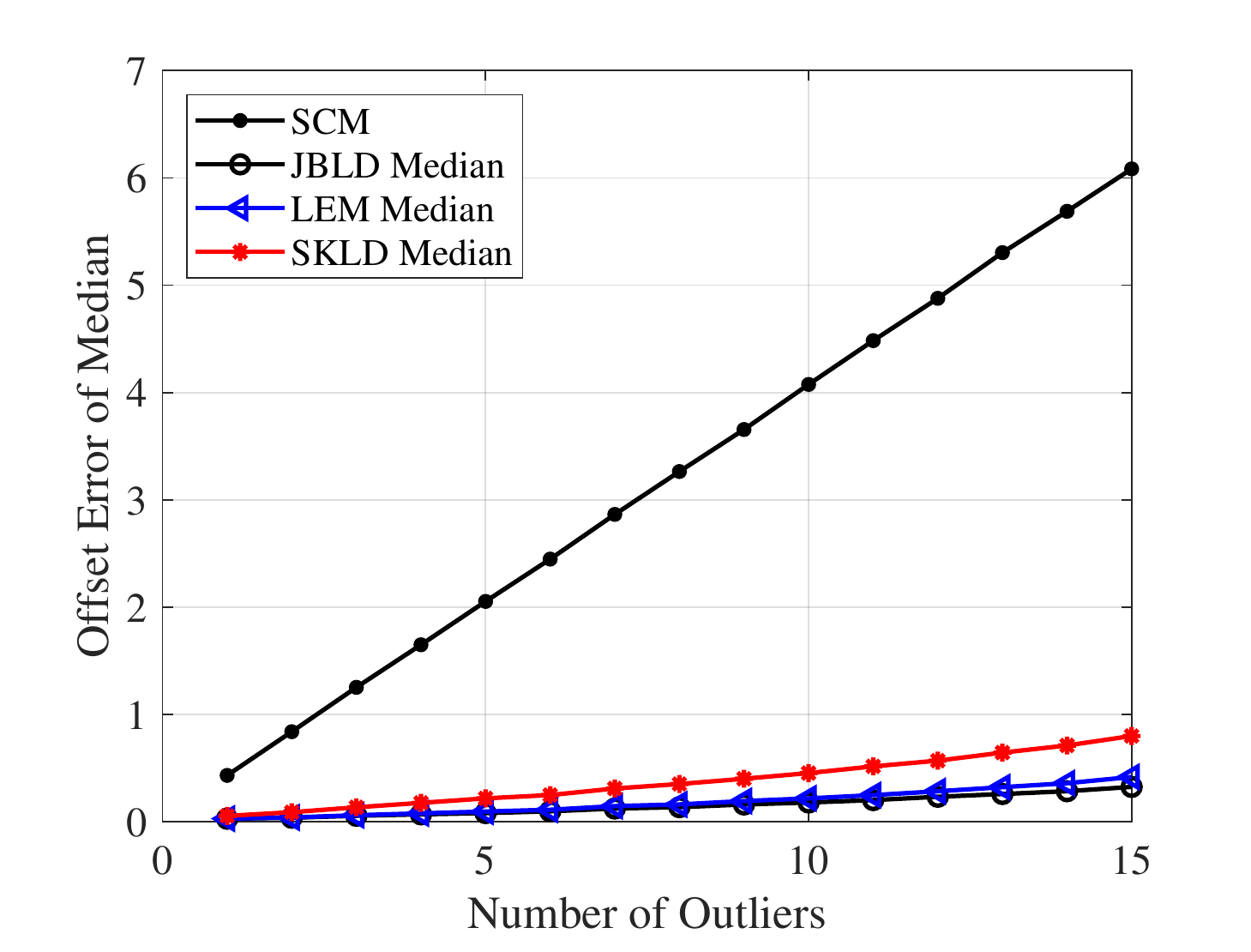}\\
  \caption{Offset errors under different number of interferences.}
  \label{Fig:Robustness_of_Num_of_Outlier}
\end{figure}

\begin{figure}[!t]
  \centering
  \includegraphics[width=8cm]{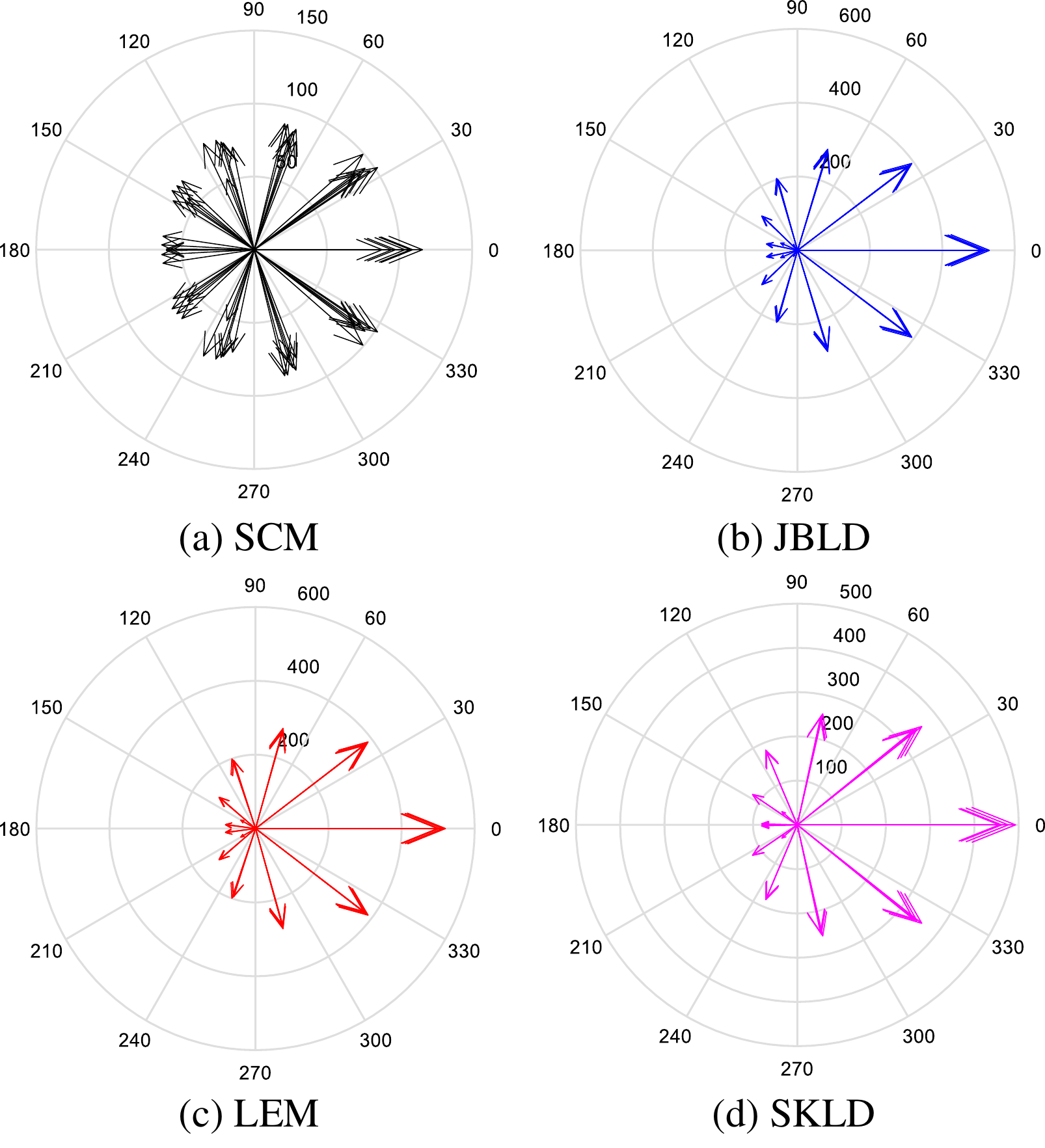}\\
  \caption{Energy distributions without any interference.}
  \label{Fig:Without_Outlier}
\end{figure}

\begin{figure}[!t]
  \centering
  \includegraphics[width=8cm]{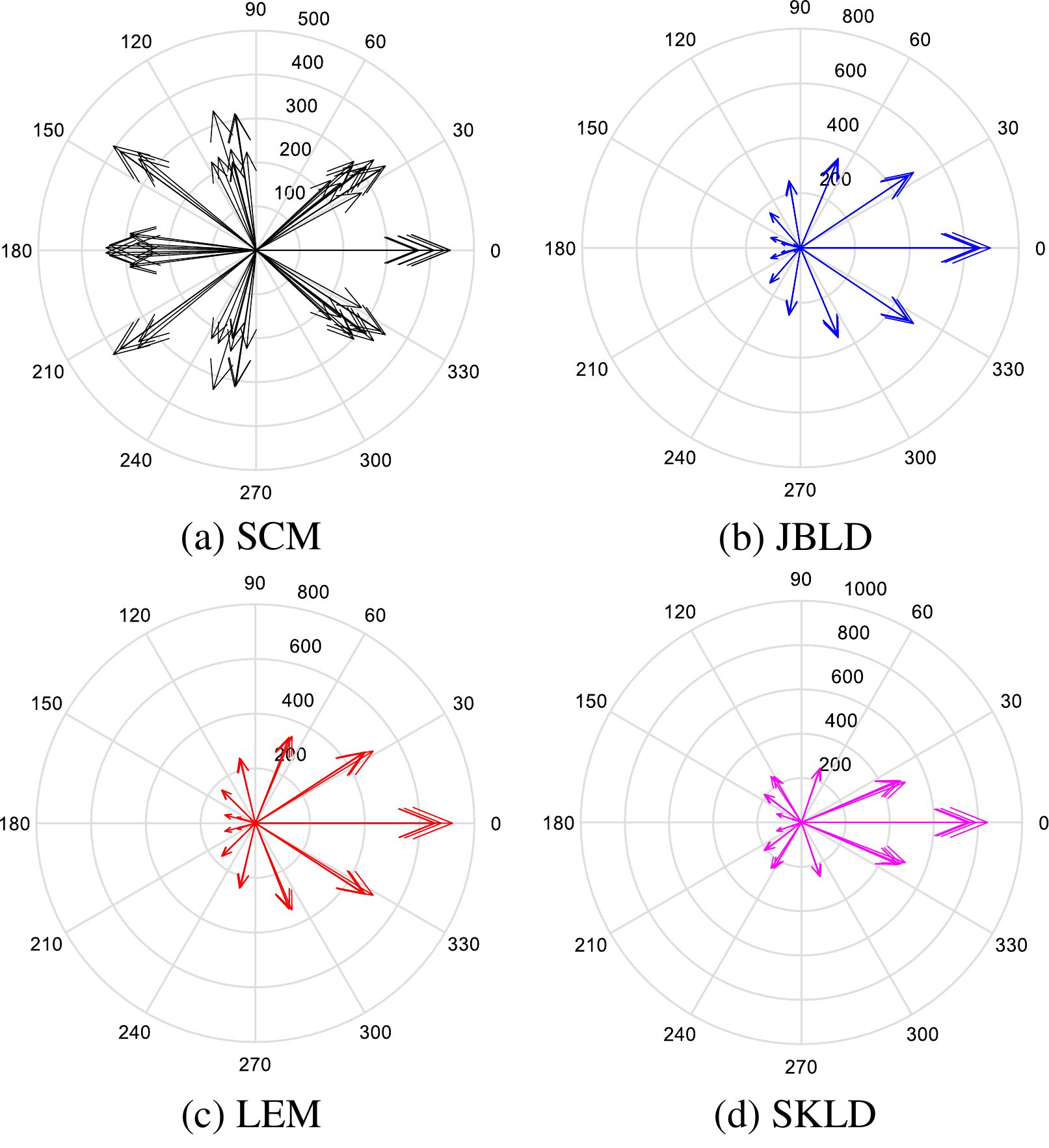}\\
  \caption{Energy distributions with $10$ interferences.}
  \label{Fig:Outlier_10}
\end{figure}

We further analyze their robustness through energy distribution analysis based on either of the following criteria: 1) a method is more robust if smaller change about its energy distribution occurs w.r.t. the same number of interferences;  2) a method is more robust if its energy distribution is more compact, namely with smaller covariance. Fig. \ref{Fig:Without_Outlier} and Fig. \ref{Fig:Outlier_10} show the energy distributions without and with $10$ interferences, respectively. From criterion 1), the JBLD median and the LEM median are  the most robust, while the SKL median is less robust but more robust compared with the SCM. From criterion 2), it is obvious that the geometric medians are more compact and hence more robust than the SCM. In particular, the JBLD median and the LEM median are again the most robust. In summary, energy distribution analysis shows that the geometric medians are more robust than the SCM and the JBLD and LEM medians are the most robust.

\begin{figure}[!t]
  \centering
  \includegraphics[width=8.5cm]{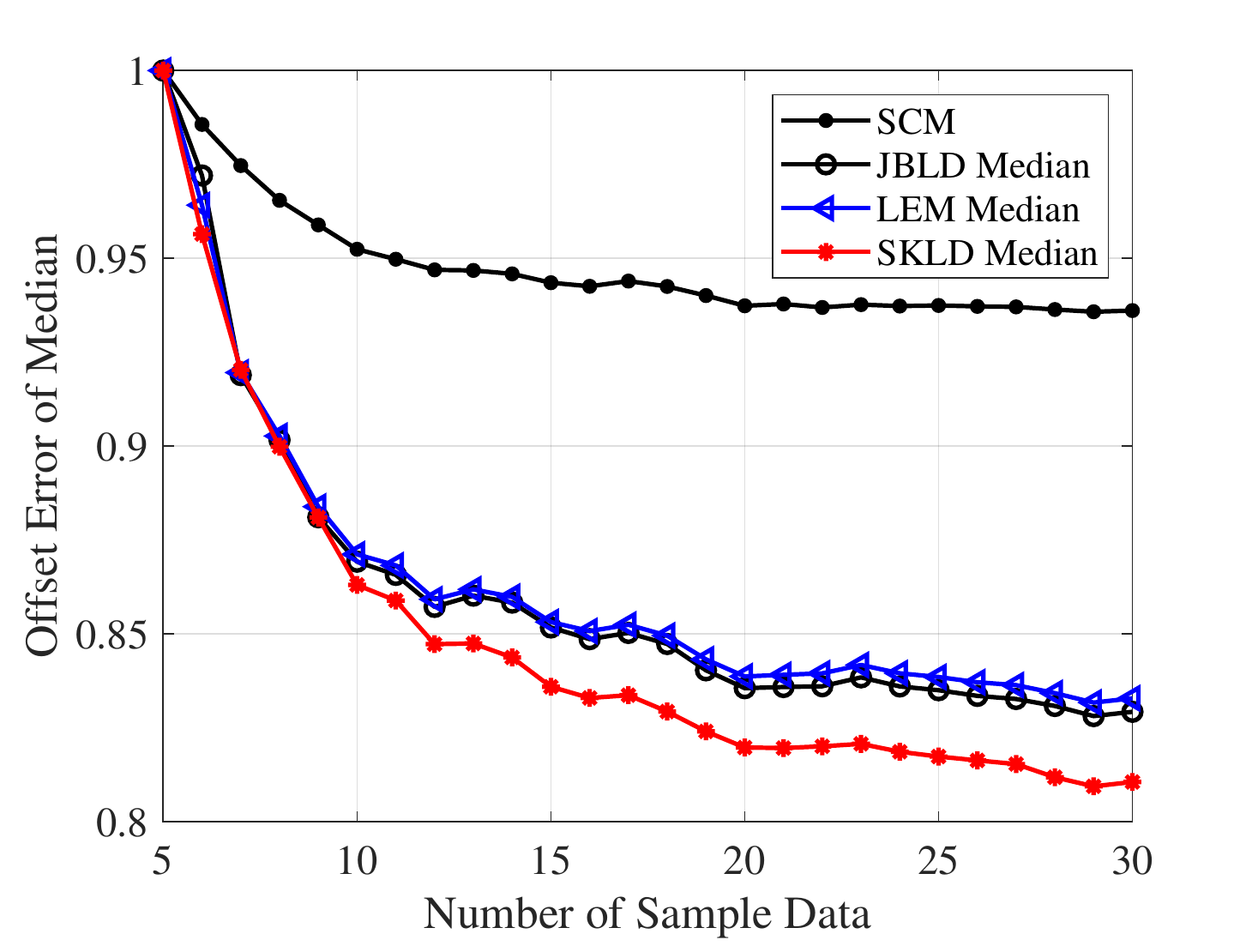}\\
  \caption{Offset errors of the SCM and geometric medians under different number of sample data.}
  \label{Fig:Robustness_of_Num_of_Sample}
\end{figure}

\begin{figure}[!t]
  \centering
  \includegraphics[width=8cm]{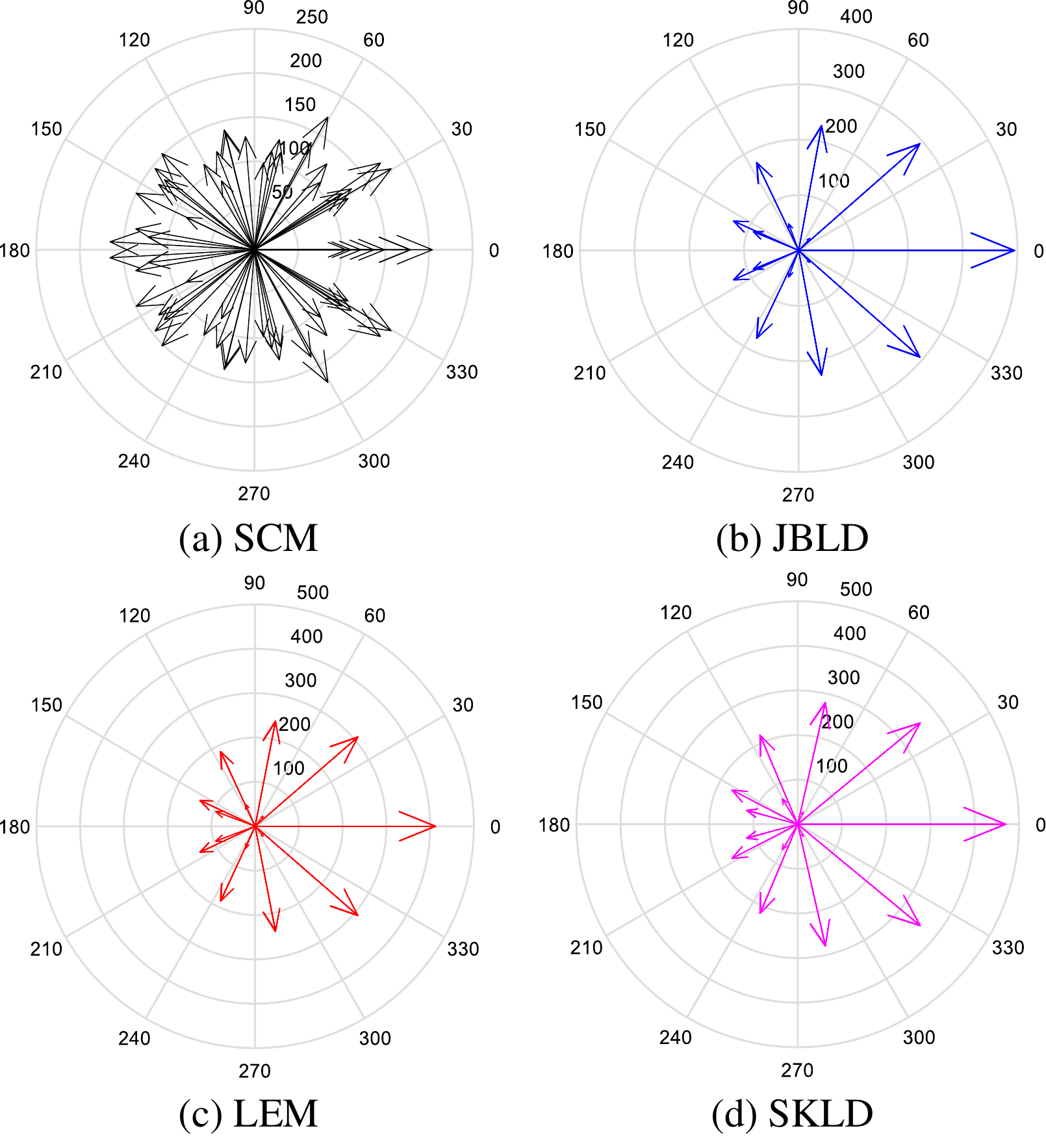}\\
  \caption{Energy distributions with $8$ sample data.}
  \label{Fig:Sample_8}
\end{figure}

Next, we are going to  show their robustness against the number of sample data. The sample dataset is generated  according to a multivariate complex Gaussian distribution with zero mean and the covariance matrix.  Denote the covariance matrix of those samples by $\bm{\widehat{R}}$, and define the error against the number of sample data by
\begin{equation}
T_{error} = \frac{\norm{ \bm{\widehat{R}} - \bm{\Sigma} }}{\norm{ \bm{\Sigma} }}.
\end{equation}

It is shown in Fig. \ref{Fig:Robustness_of_Num_of_Sample} that the SKL median is the most robust and geometric medians are much more robust than the SCM.

In Fig. \ref{Fig:Sample_8}, we show the energy distributions with $8$ sample data. It is clear that the geometric medians are more compact and hence more robust compared with the SCM.

\subsection{The discriminative power}

In this subsection, we analyze the discriminative power by comparing the statistics with and without manifold filter. First of all, we compare the normalized detection statistics before and after weighting the matrices; see Eq. \eqref{eq:manitomani} in Section \ref{sec:mantoman}. A number of $40$ sample data is simulated and an interference is injected into the $20$-th sample with $\textrm{SCNR} = 15$ {dB} and $f_d = 0.2$. Parameters $m$ and $h$ are chosen to be $m=11,h=1.5$, $m=11,h=2$ and $m=13,h=1.5$, respectively. The detection statistic, namely the distance between the CMM and the HPD matrix in the CUT, is computed with and without manifold filter. The normalized detection statistics are plotted in Fig. \ref{Fig:Statistics}. Obviously, the detection statistic after manifold filter is smaller than that before filter in the cell without target signal. In other words, after manifold filter, the distance between clutter and clutter becomes smaller while the distance between target and clutter becomes larger.

Anisotropy analysis is also used to illustrate the discriminative power. To analyze the difference between the target signal and the CCM in terms of the anisotropy on the matrix manifold, we define an anisotropy index of discrimination (AD) to measure the difference on the matrix manifold as
\begin{equation}
\textrm{AD} = \frac{\textrm{AI}_{signal}}{\textrm{AI}_{CCM}}.
\end{equation}
Here, $\textrm{AI}_{signal}$ and $\textrm{AI}_{CCM}$ denote the target and clutter anisotropies, respectively reflecting the local geometric structures at locations of the target and the CCM on the matrix manifold. $\textrm{AD}$ represents the difference in the anisotropy for the target and the CCM. The larger the value of $\textrm{AD}$, the greater the difference between the two points. Fig. \ref{Fig:Detection_Potential} shows the $\textrm{AD}$ with $100$ Monte Carlo trials. It can be seen that after manifold filter, the values of $\textrm{AD}$ become larger, implying that the difference in local geometric structures between target signal and clutter becomes larger after manifold filter.

\begin{figure}[!t]
  \centering
  \subfigure[JBLD divergence]{\includegraphics[width=7.5cm]{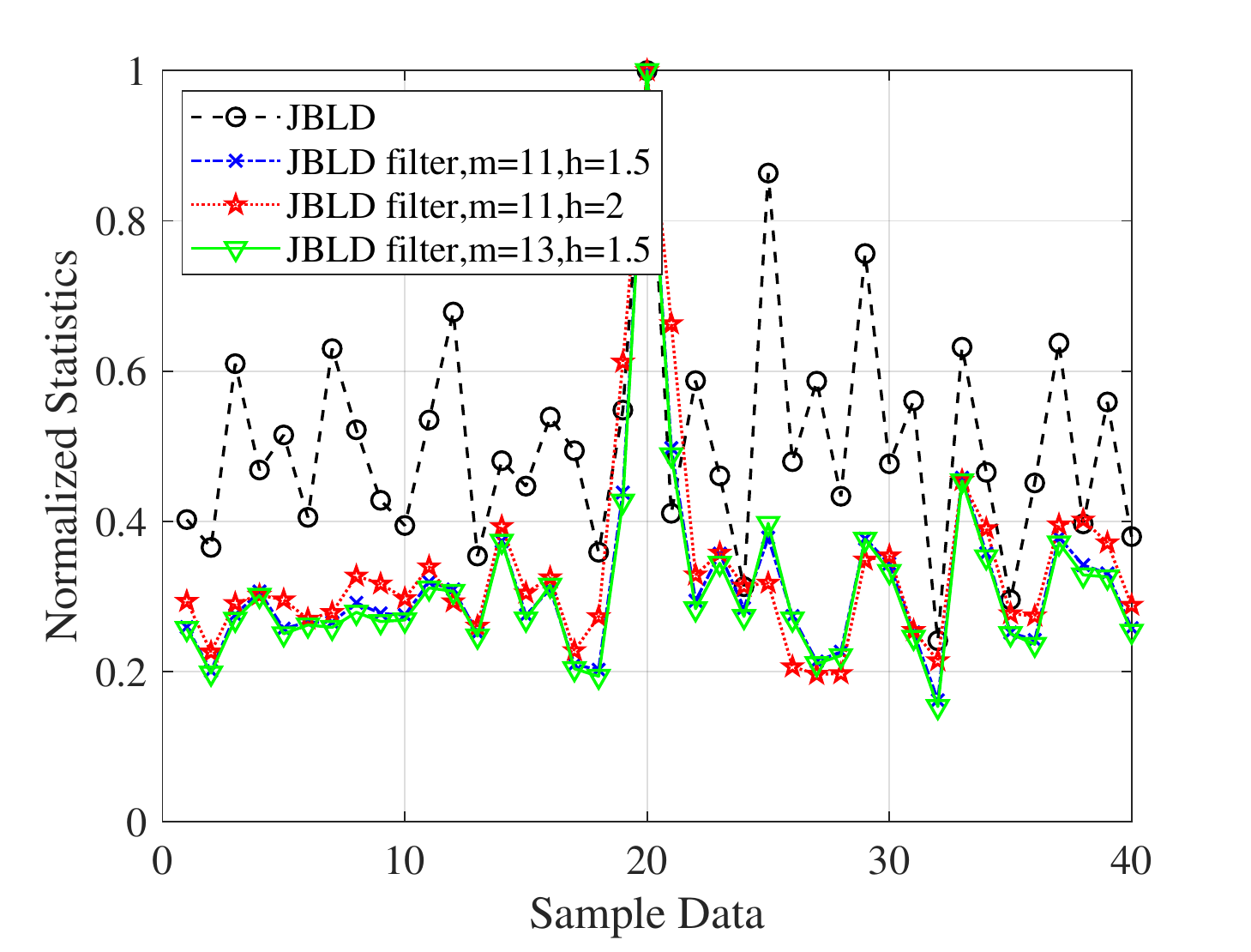}}\\
  \subfigure[LEM]{\includegraphics[width=7.5cm]{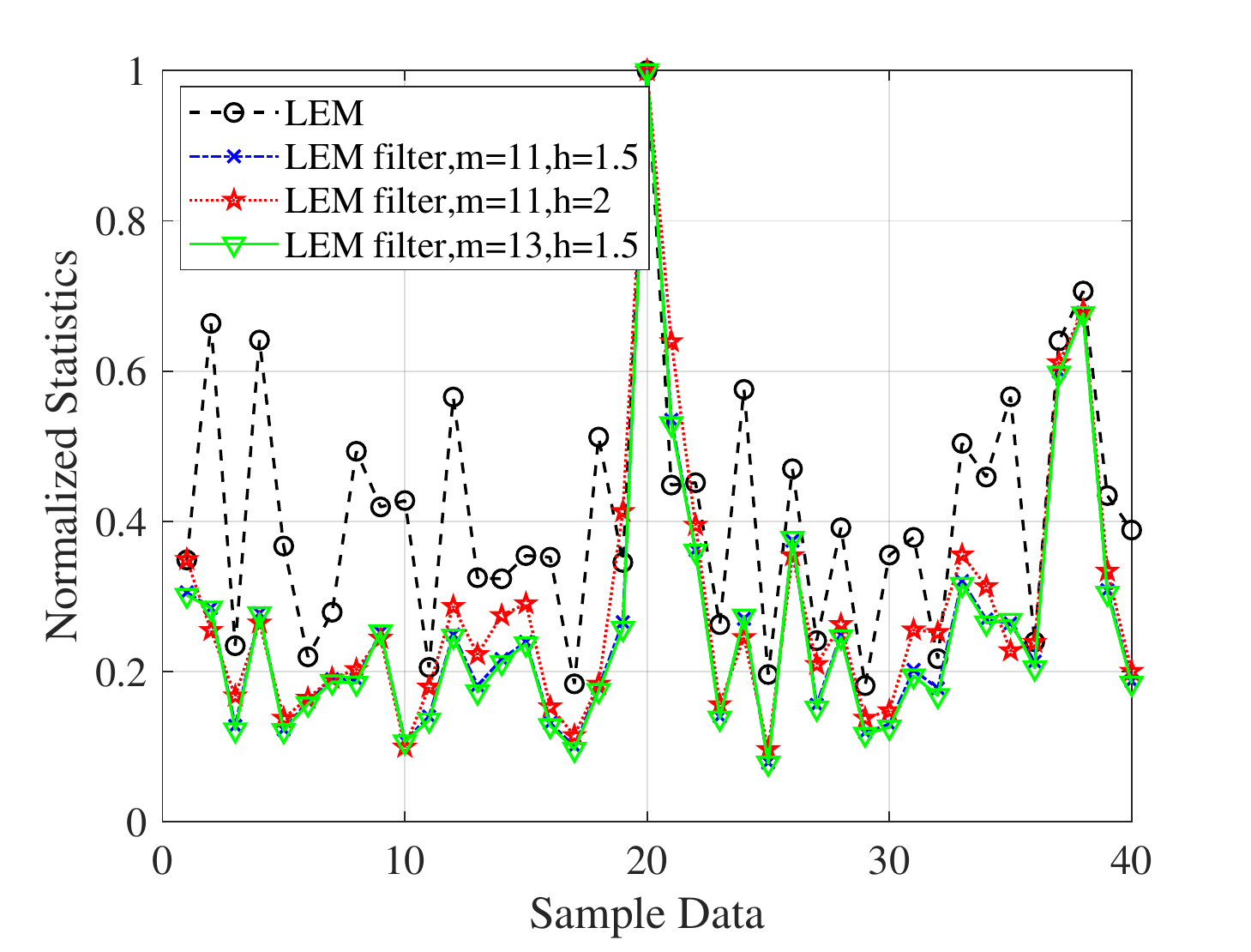}}\\
  \subfigure[SKLD]{\includegraphics[width=7.5cm]{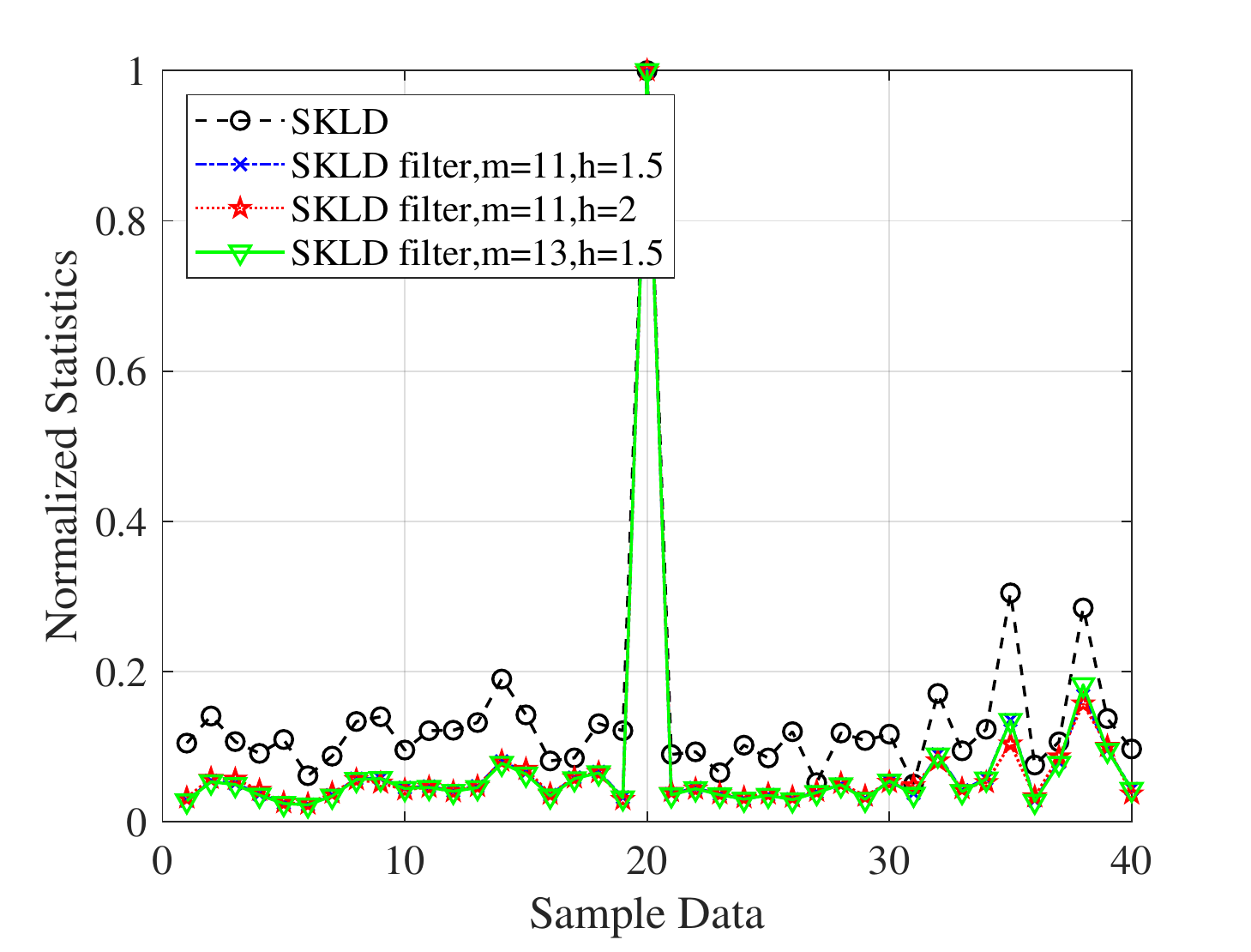}}\\
  \caption{Normalized statistics with and without manifold filter}
  \label{Fig:Statistics}
\end{figure}

\begin{figure}[!t]
  \centering
  \includegraphics[width=9cm]{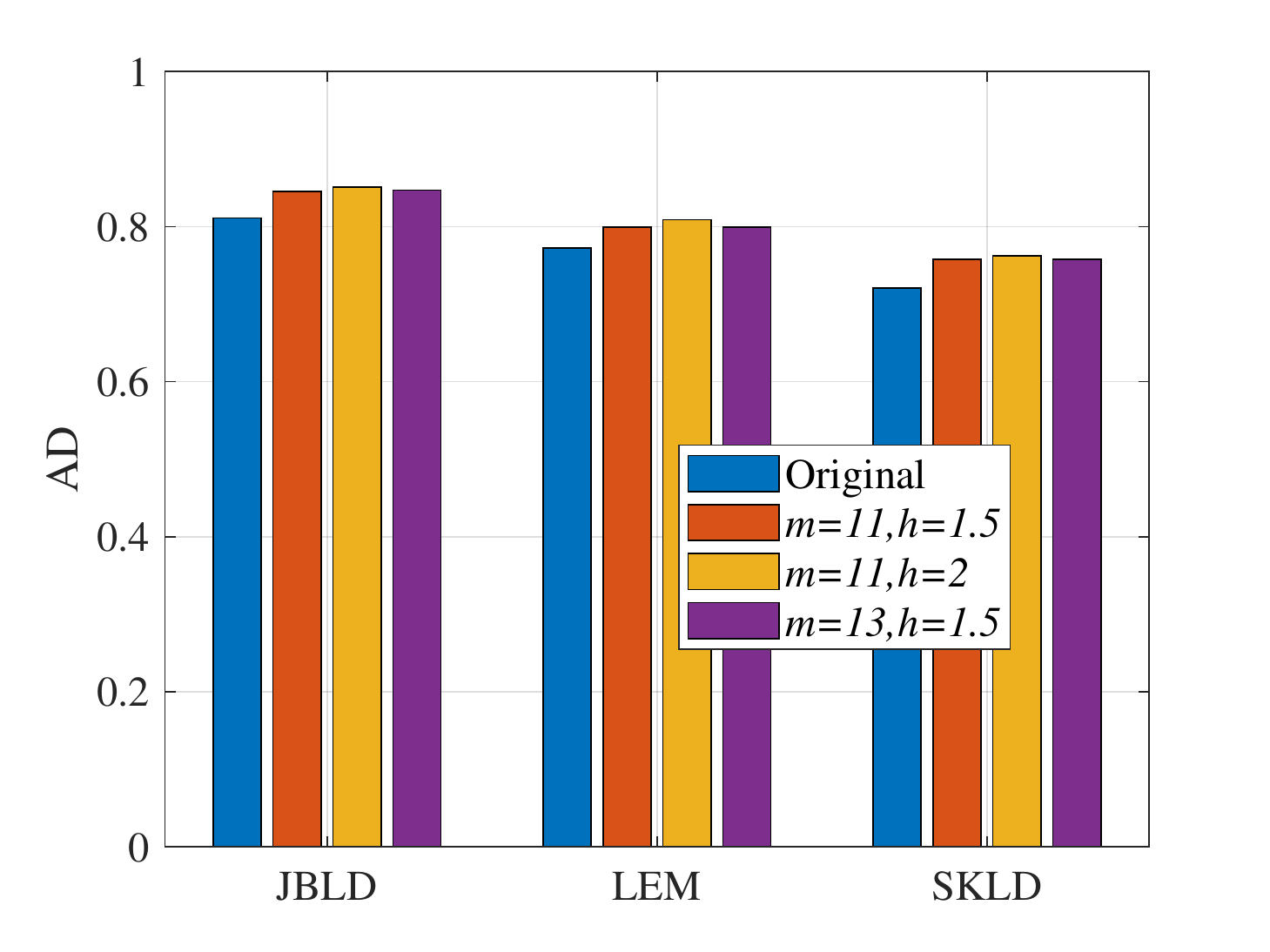}\\
  \caption{Discriminative power.}
  \label{Fig:Detection_Potential}
\end{figure}

\subsection{Detection performances}

In this subsection, numerical simulations are conducted to validate  performances of the proposed detectors in comparison with their counterparts and the AMF. To ease from huge computational load, the probability of false alarm $P_{FA}$ is chosen to be $10^{-3}$. For the Monte Carlo simulations, the detection threshold and the detection probability $P_d$ are obtained from $100/P_{FA}$ and $2000$ independent trials, respectively. In the following simulations, the dimension of matrices is $n=8$; two interferences are randomly added to the supplementary data with $\textrm{SCNR} = 10$ dB and the normalized frequency $f_d = 0.2$. Figs. \ref{JBLD_Pd}, \ref{LEM_Pd} and \ref{SKL_Pd} show the detection performances of the proposed MIG median  detectors in Gaussian and non-Gaussian clutter with two interferences. The number of sample data is $K=8$ or $12$. When $K=8$, the AMF is much worse compared with the MIG median detectors with or without manifold filter. Moreover, the MIG median  detectors with manifold filter have better performance than the MIG median  detectors without manifold filter. When $K=1.5n$, the AMF and the geometric median detectors behave very similarly while geometric median detectors with manifold filter are slightly better than the AMF.

\begin{figure*}[ht]
\centering
\subfigure[Gaussian, $K=8, n=8$] {\includegraphics[width=8cm,angle=0]{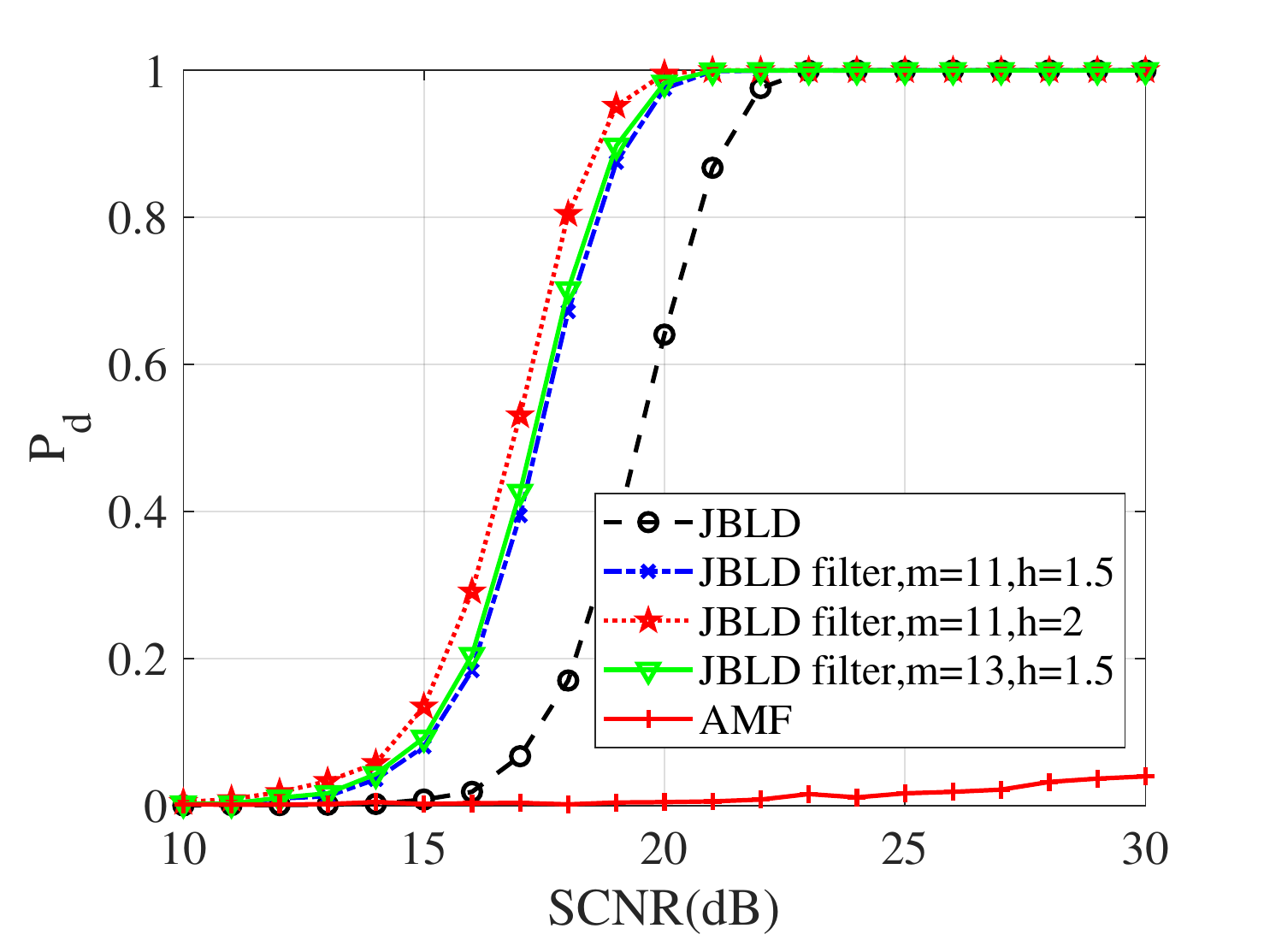}}
\subfigure[Gaussian, $K=12, n=8$] {\includegraphics[width=8cm,angle=0]{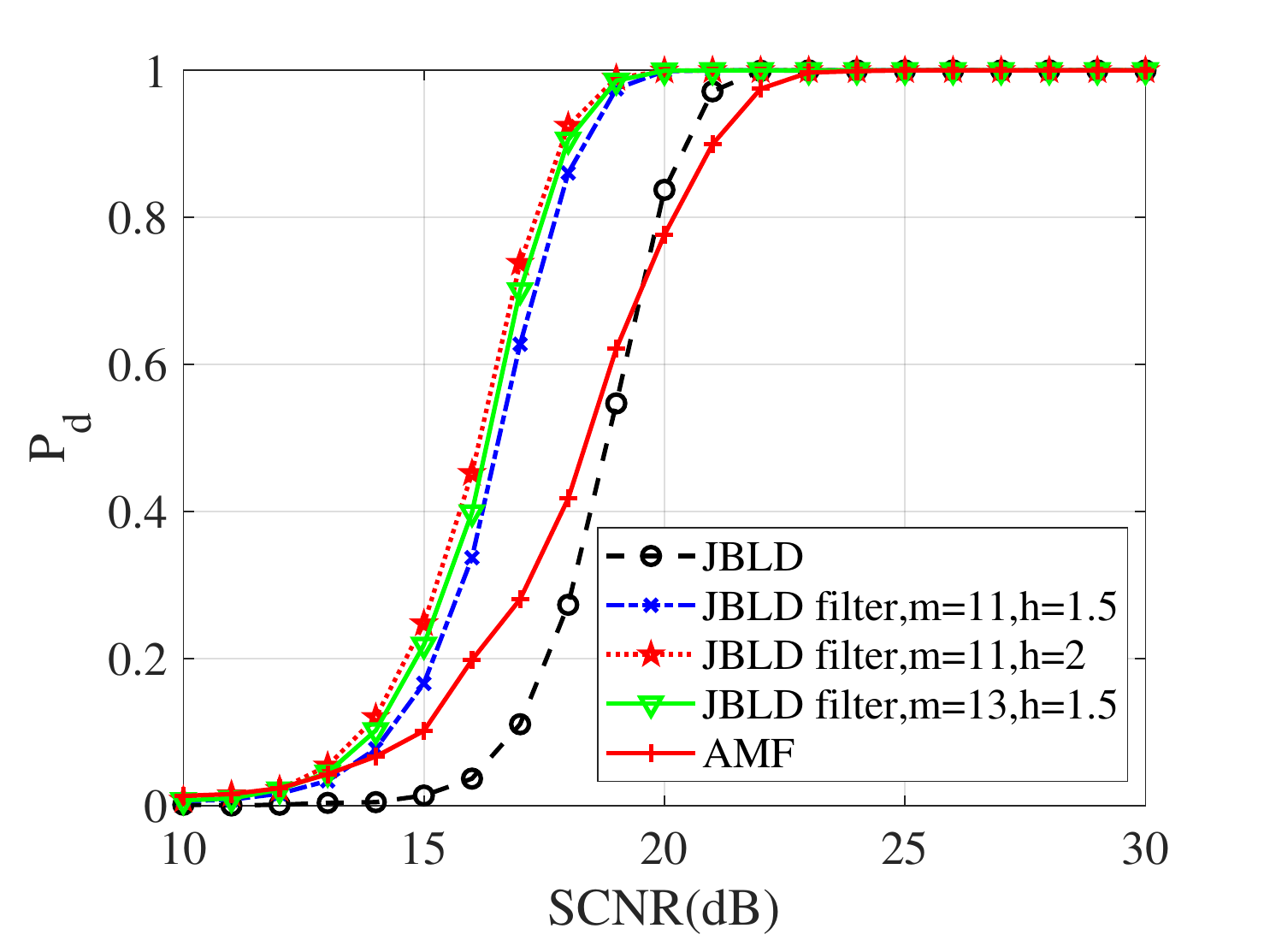}}
\subfigure[Non-Gaussian, $K=8, n=8$] {\includegraphics[width=8cm,angle=0]{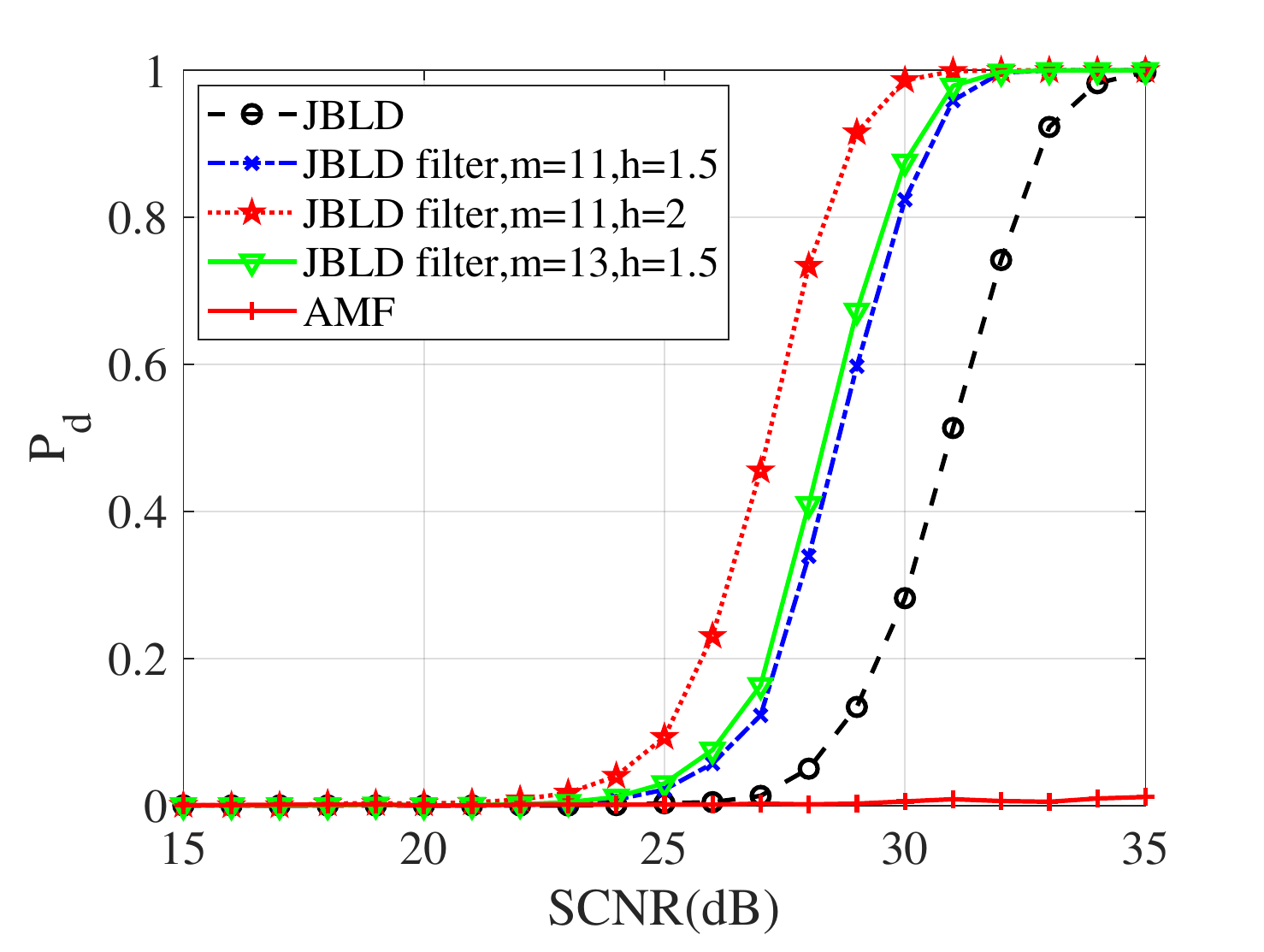}}
\subfigure[Non-Gaussian, $K=12, n=8$] {\includegraphics[width=8cm,angle=0]{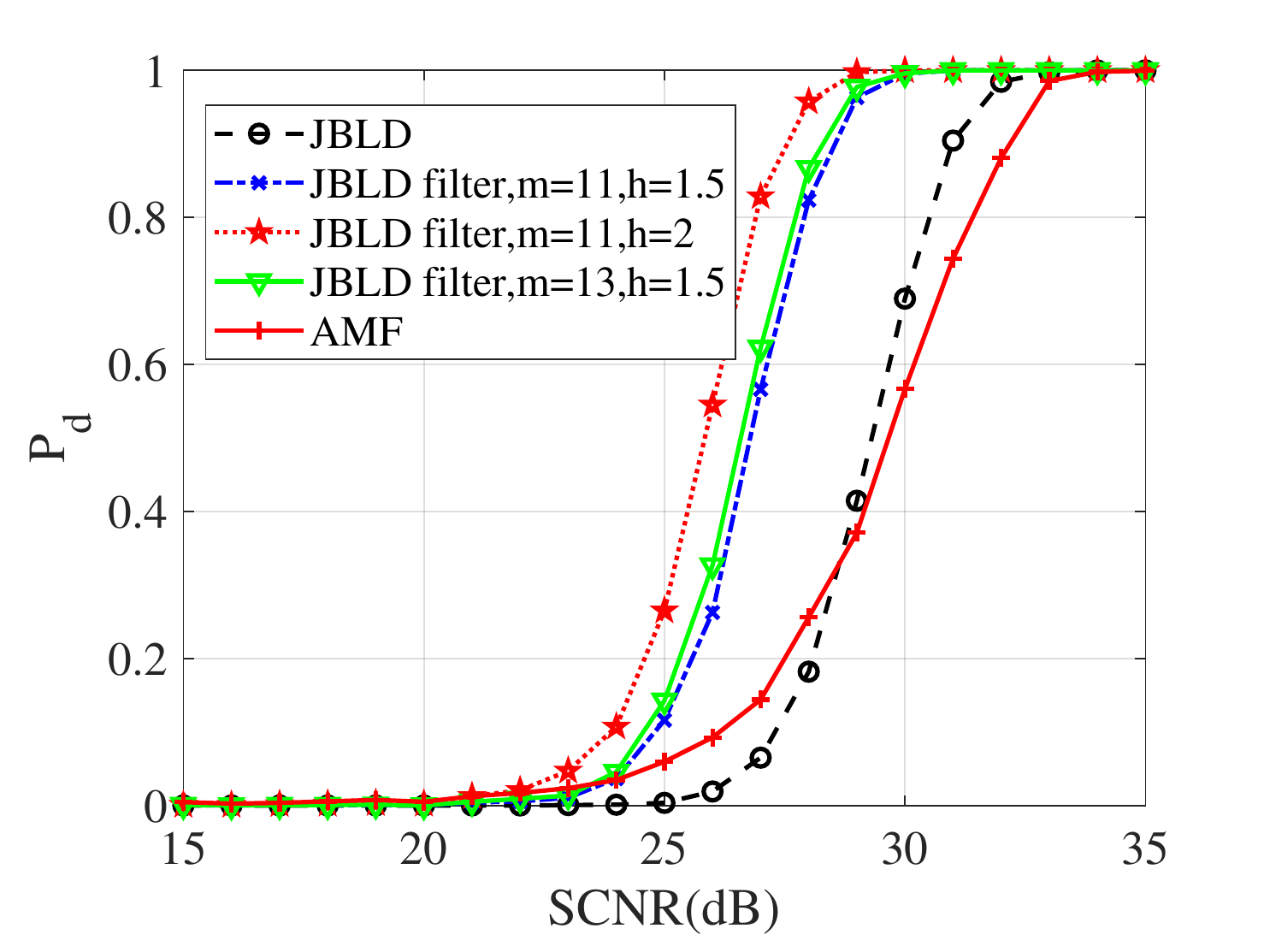}}
\caption{$P_d$ vs SCNR with JBLD divergence}
\label{JBLD_Pd}
\end{figure*}

\begin{figure*}[ht]
\centering
\subfigure[Gaussian, $K=8, n=8$] {\includegraphics[width=8cm,angle=0]{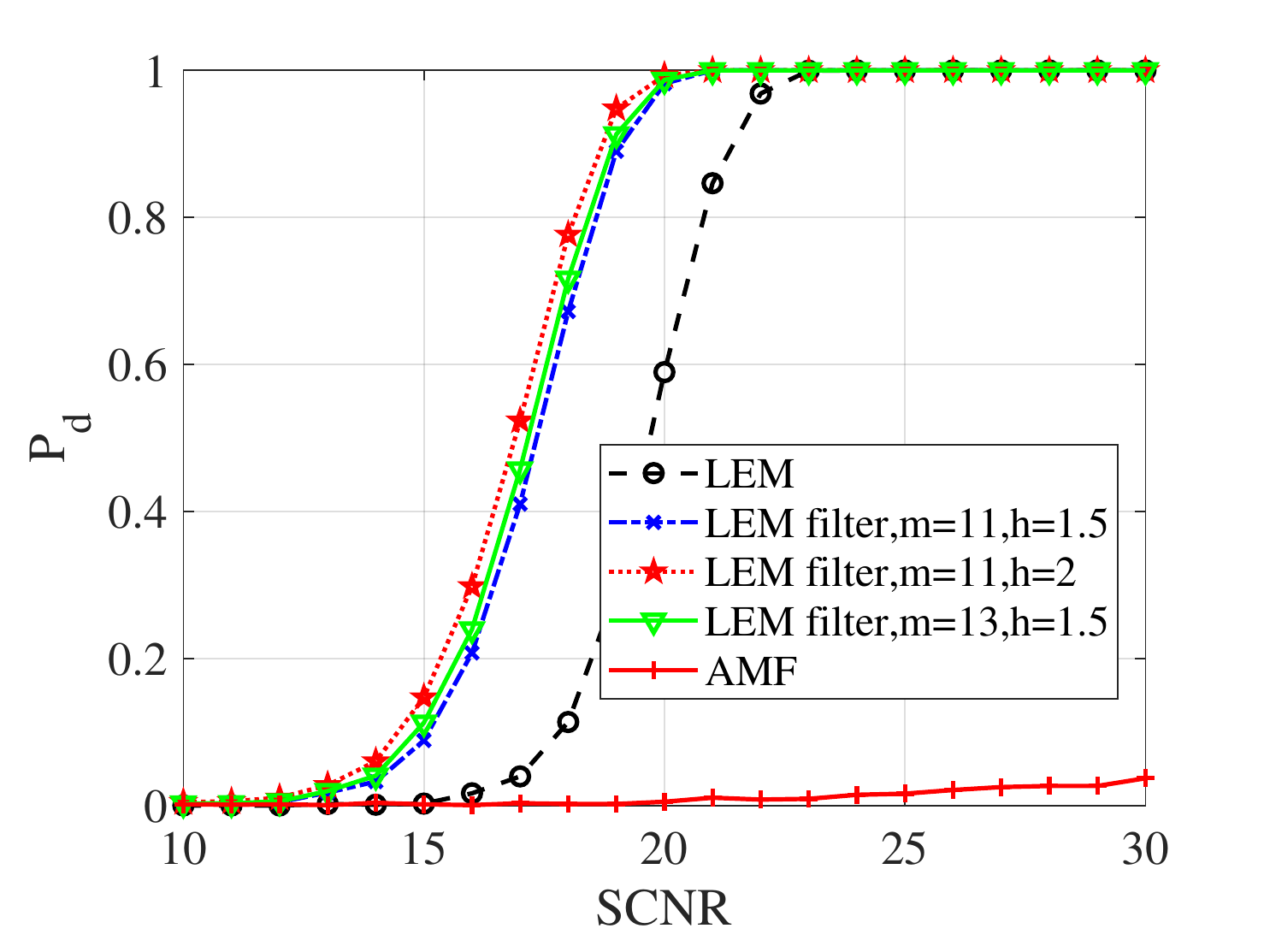}}
\subfigure[Gaussian, $K=12, n=8$] {\includegraphics[width=8cm,angle=0]{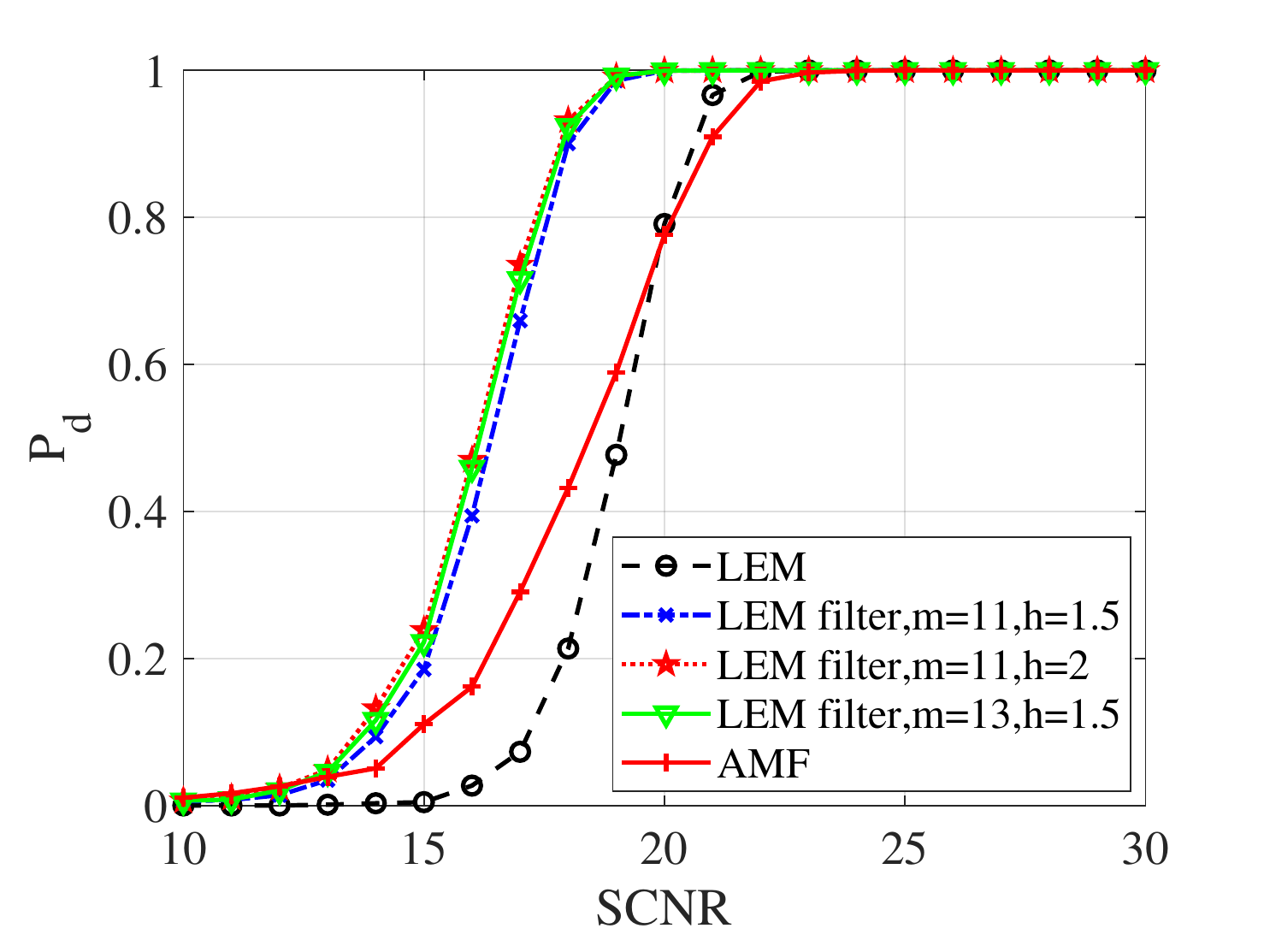}}
\subfigure[Non-Gaussian, $K=8, n=8$] {\includegraphics[width=8cm,angle=0]{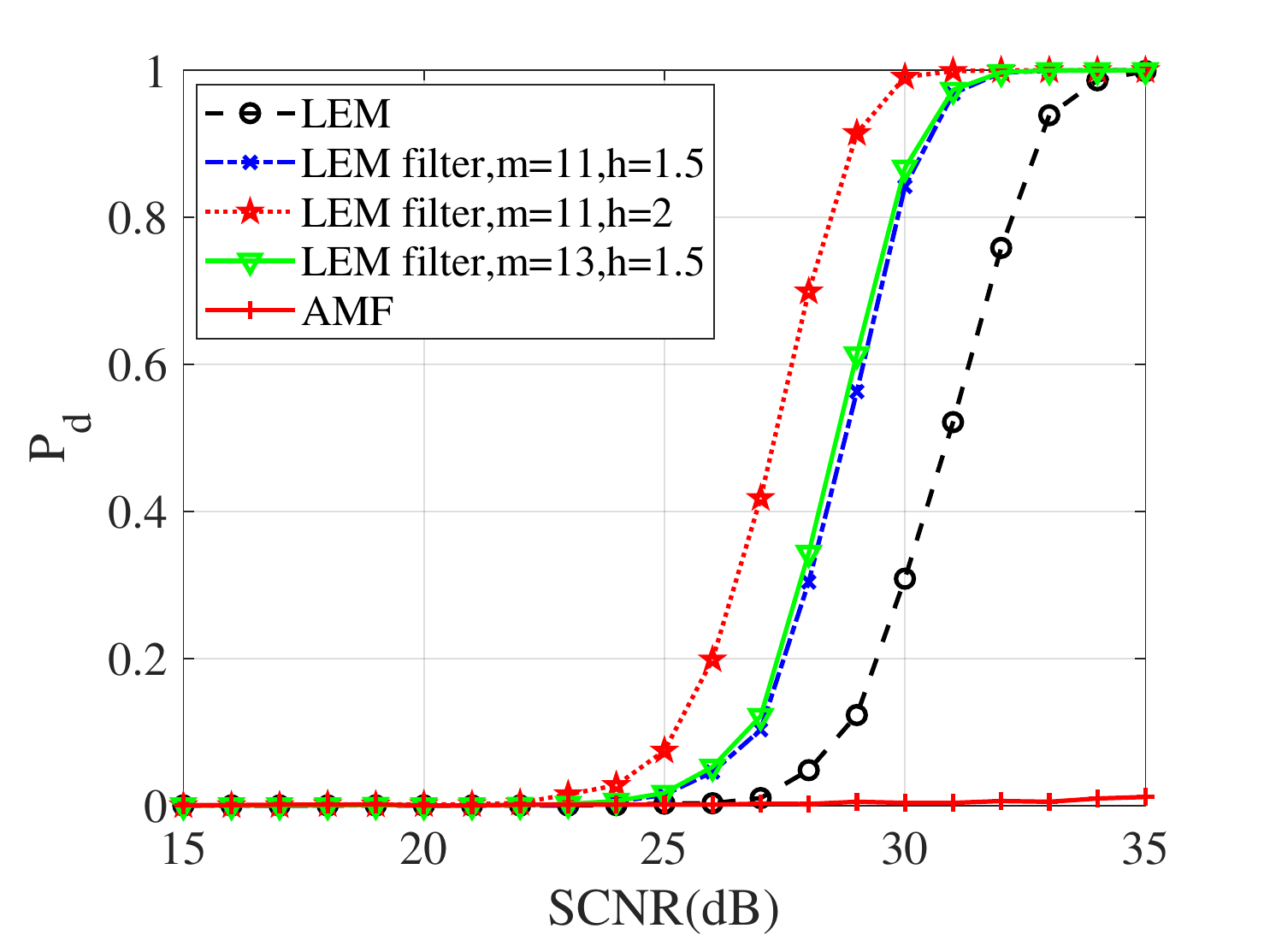}}
\subfigure[Non-Gaussian, $K=12, n=8$] {\includegraphics[width=8cm,angle=0]{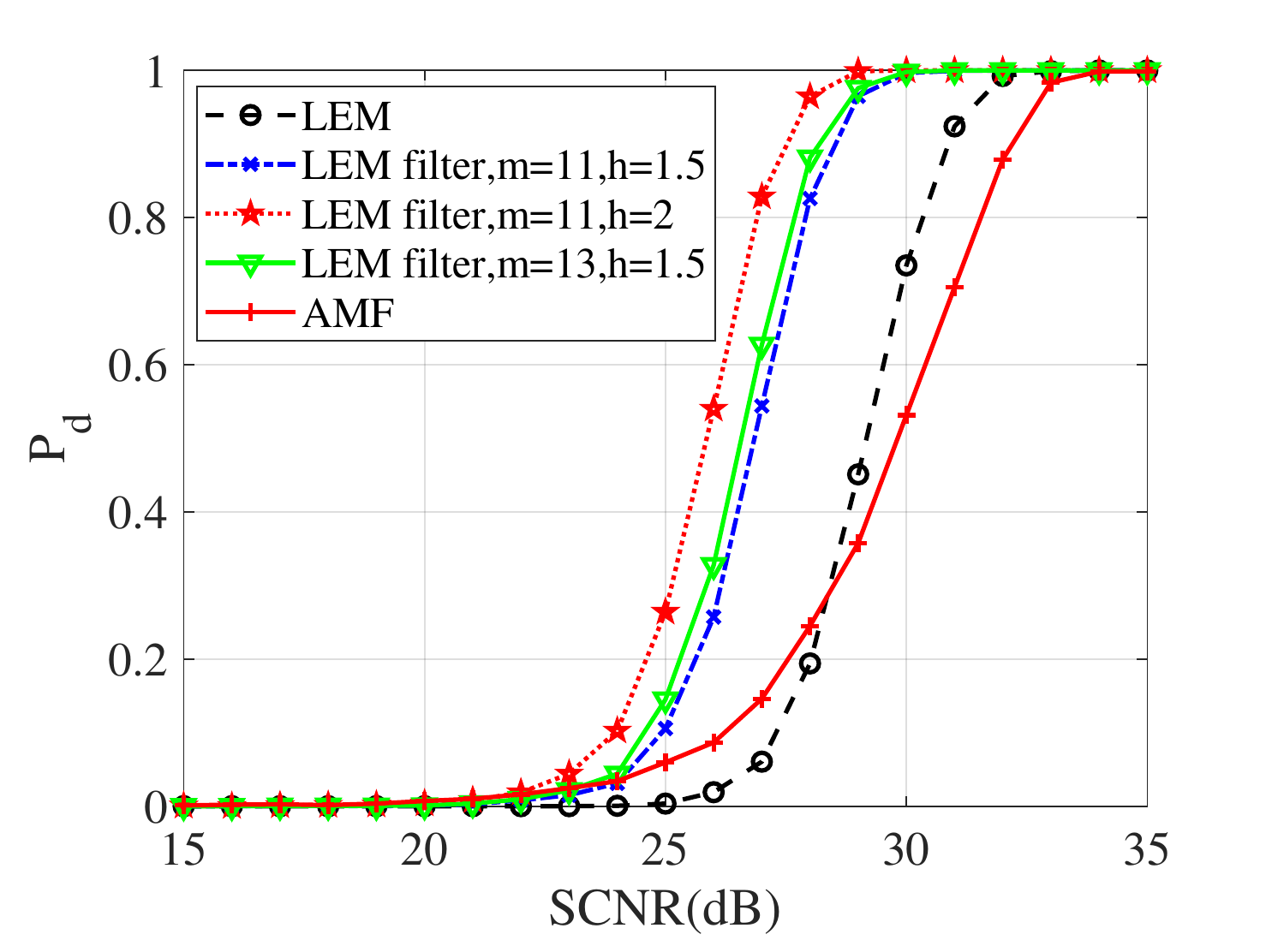}}
\caption{$P_d$ vs SCNR with LEM}
\label{LEM_Pd}
\end{figure*}

\begin{figure*}[ht]
\centering
\subfigure[Gaussian, $K=8, n=8$] {\includegraphics[width=8cm,angle=0]{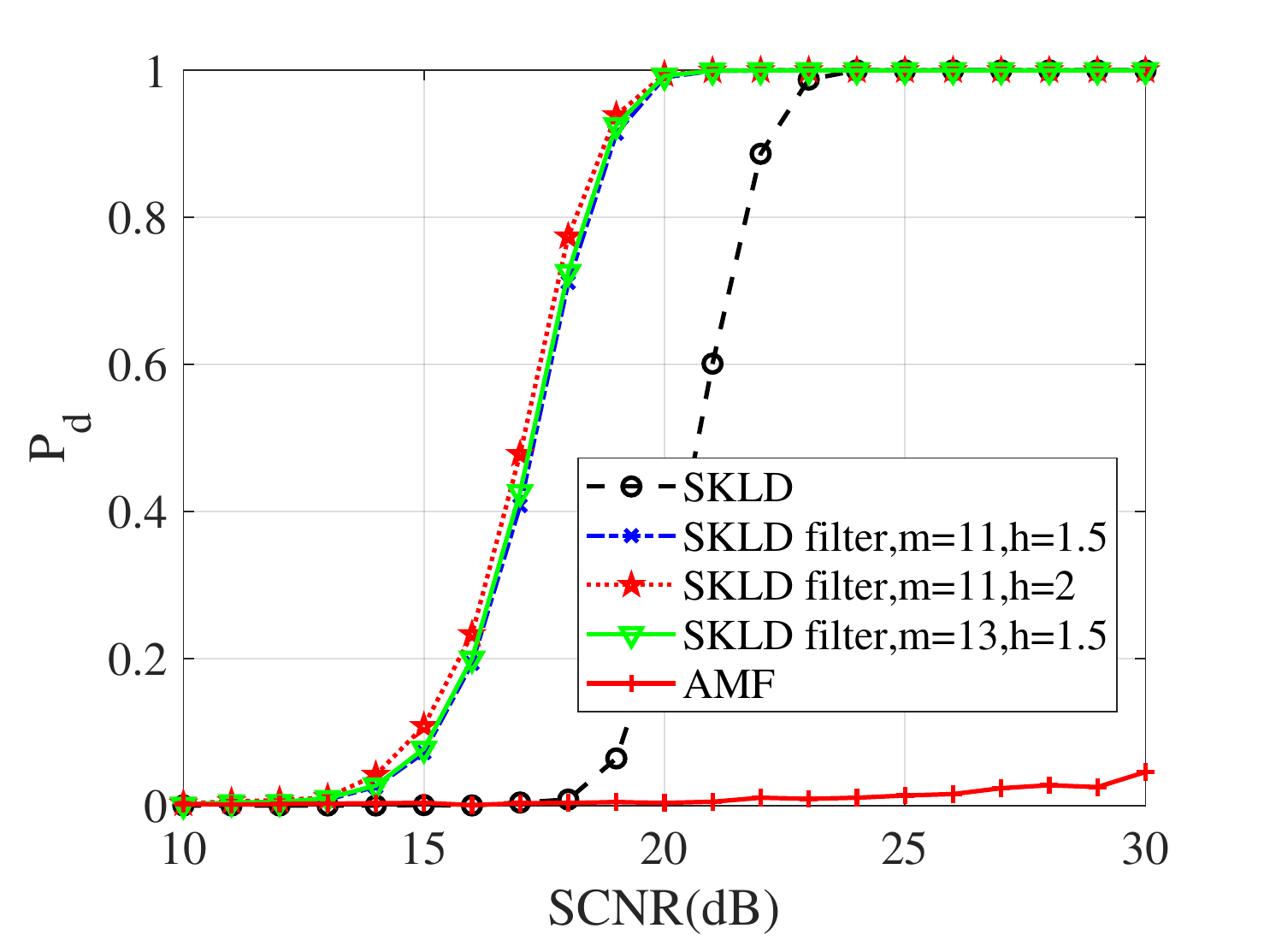}}
\subfigure[Gaussian, $K=12, n=8$] {\includegraphics[width=8cm,angle=0]{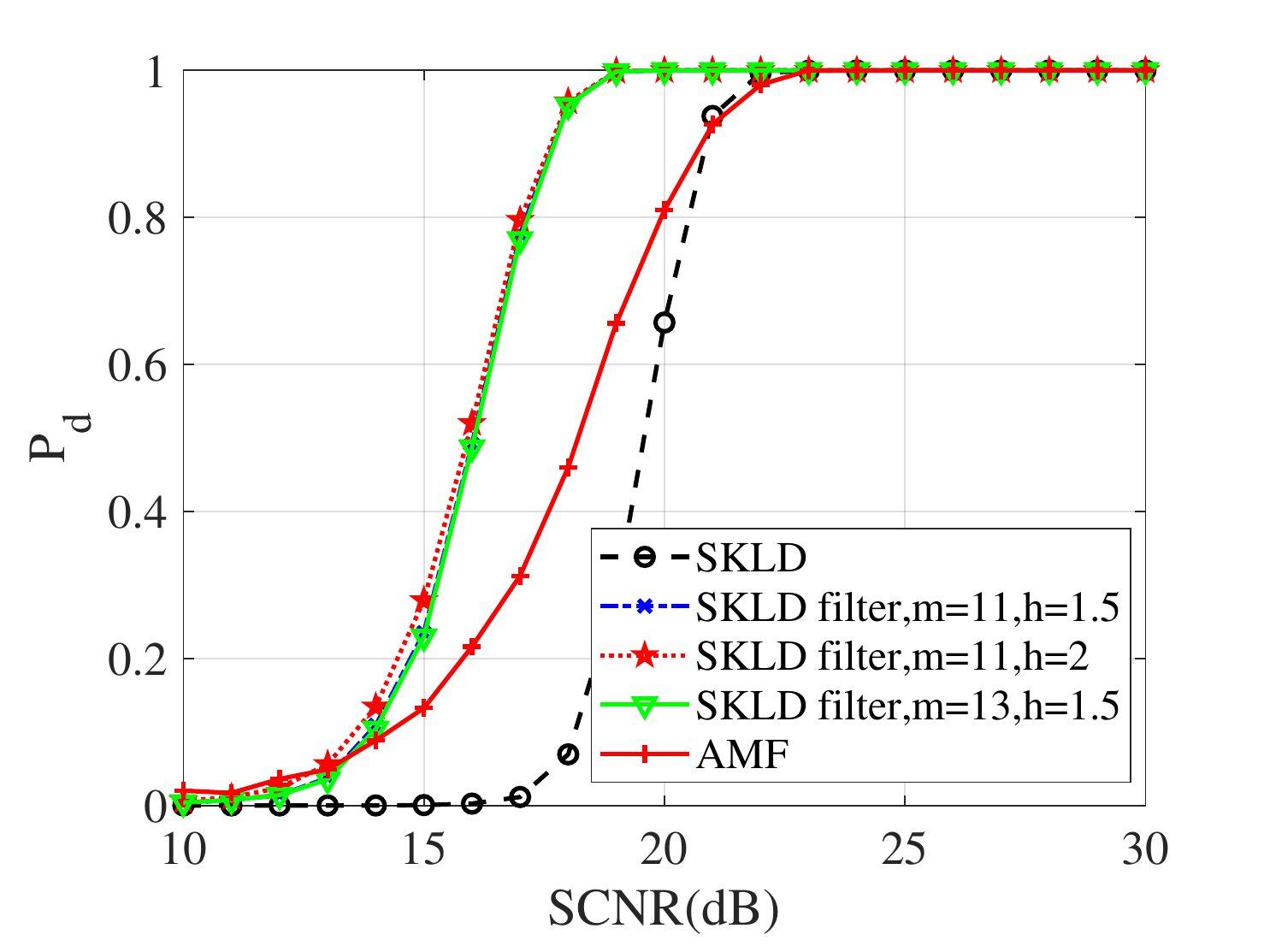}}
\subfigure[Non-Gaussian, $K=8, n=8$] {\includegraphics[width=8cm,angle=0]{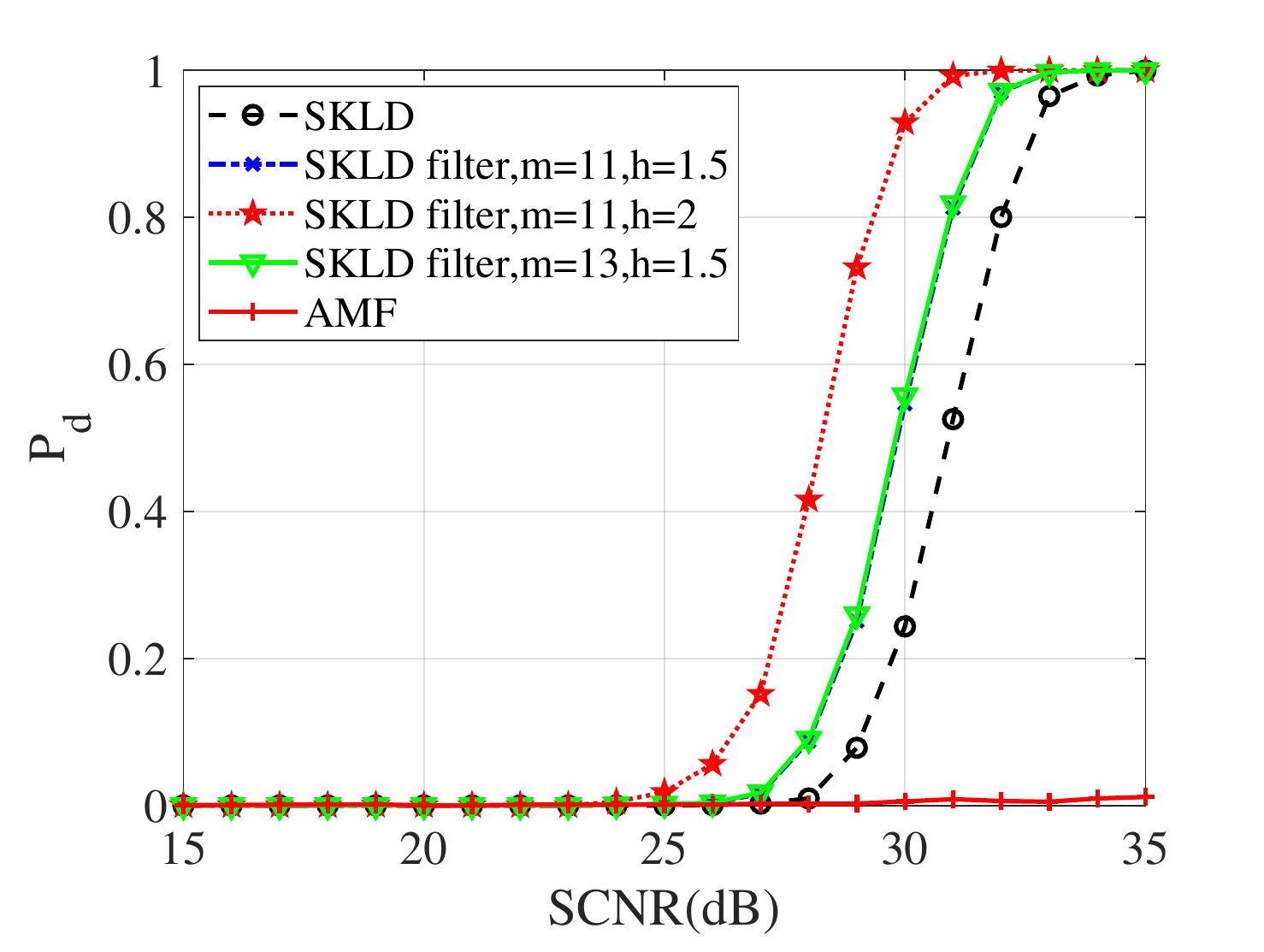}}
\subfigure[Non-Gaussian, $K=12, n=8$] {\includegraphics[width=8cm,angle=0]{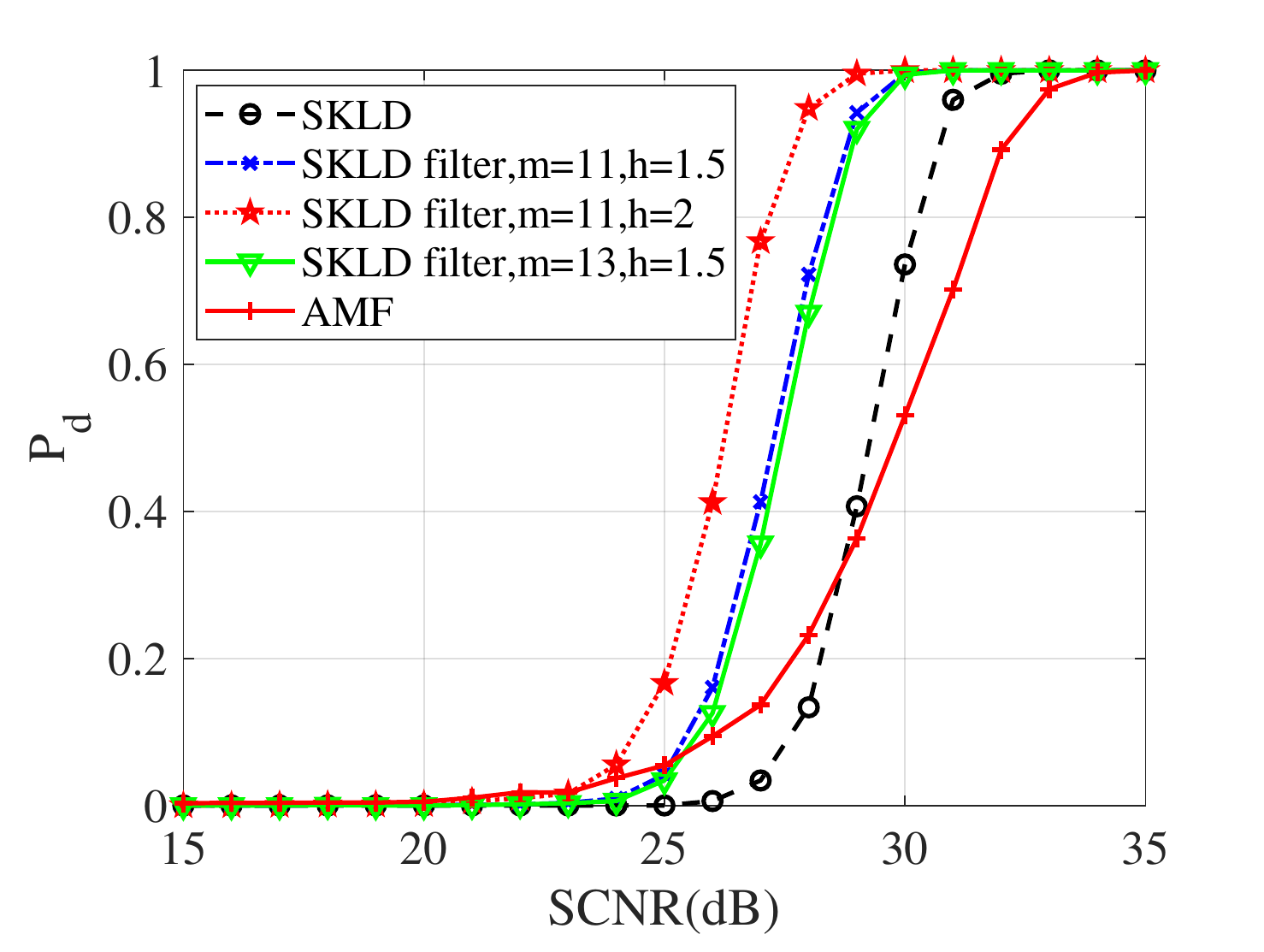}}
\caption{$P_d$ vs SCNR with SKLD}
\label{SKL_Pd}
\end{figure*}

\section{Conclusions}
\label{sec:con}

In this paper, we designed a class of MIG median  detectors with manifold filter to detect signals embedded in nonhomogeneous environments. The sample data was modeled as an HPD matrix with the Toeplitz structure. Particularly, the clutter covariance matrix was estimated as a geometric median for a set of HPD matrices w.r.t. various geometric measures. Then, the problem of signal detection was formulated into the discrimination of two points on the HPD matrix manifold. Furthermore, a manifold-to-manifold map was defined by using the manifold filter to alter the manifold structure, that consequently improved the discrimination power on the matrix manifold. Numerical simulations showed the performance advantages of those MIG median detectors compared with the AMF and their state-of-the-art counterparts.  
Similar to the AMF, MIG detectors can also be widely applied in practical problems, for instance, target detection in sea or ground clutter, sonar signal detection, and signal detection in communication system.  Potential future research will concern the subspace signal detection and their applications for range distributed target detection and synthetic aperture radar image target detection.

\section*{Acknowledgements}
This work was partially supported by the National Natural Science Foundation of China under Grant No. 61901479 and JST-CREST. L. Peng  is  an adjunct faculty member of Waseda Institute for Advanced Study, Waseda University, Japan, and School of Mathematics and Statistics, Beijing Institute of Technology, China.


\bibliographystyle{plain}
\bibliography{mybibfile}

\begin{thebibliography}{10}

\bibitem{8792393}
A.~{Ali}, N.~{Gonz\'alez-Prelcic}, and R.~W. {Heath}.
\newblock Spatial covariance estimation for millimeter wave hybrid systems
  using out-of-band information.
\newblock {\em IEEE Transactions on Wireless Communications},
  18(12):5471--5485, 2019.

\bibitem{6514112}
M.~{Arnaudon}, F.~{Barbaresco}, and L.~{Yang}.
\newblock Riemannian medians and means with applications to radar signal
  processing.
\newblock {\em IEEE Journal of Selected Topics in Signal Processing},
  7(4):595--604, 2013.

\bibitem{8052571}
A.~{Aubry}, A.~{De Maio}, and L.~{Pallotta}.
\newblock A geometric approach to covariance matrix estimation and its
  applications to radar problems.
\newblock {\em IEEE Transactions on Signal Processing}, 66(4):907--922, 2018.

\bibitem{6166345}
A.~{Aubry}, A.~{De Maio}, L.~{Pallotta}, and A.~{Farina}.
\newblock Maximum likelihood estimation of a structured covariance matrix with
  a condition number constraint.
\newblock {\em IEEE Transactions on Signal Processing}, 60(6):3004--3021, 2012.

\bibitem{6573681}
A.~{Aubry}, A.~D. {Maio}, L.~{Pallotta}, and A.~{Farina}.
\newblock Covariance matrix estimation via geometric barycenters and its
  application to radar training data selection.
\newblock {\em IET Radar, Sonar \& Navigation}, 7(6):600--614, July 2013.

\bibitem{6825699}
A.~{Aubry}, A.~D. {Maio}, L.~{Pallotta}, and A.~{Farina}.
\newblock Median matrices and their application to radar training data
  selection.
\newblock {\em IET Radar, Sonar \& Navigation}, 8(4):265--274, 2014.

\bibitem{5484507}
F.~{Bandiera}, O.~{Besson}, and G.~{Ricci}.
\newblock Knowledge-aided covariance matrix estimation and adaptive detection
  in compound-{G}aussian noise.
\newblock {\em IEEE Transactions on Signal Processing}, 58(10):5391--5396,
  2010.

\bibitem{4720937}
F.~{Barbaresco}.
\newblock {Innovative tools for radar signal processing Based on Cartan's
  geometry of SPD matrices \& Information Geometry}.
\newblock In {\em 2008 IEEE Radar Conference}, pages 1--6, 2008.

\bibitem{6450629}
F.~{Barbaresco}.
\newblock {RRP 3.0: 3rd generation Robust Radar Processing based on Matrix
  Information Geometry (MIG)}.
\newblock In {\em 2012 9th European Radar Conference}, pages 42--45, 2012.

\bibitem{9078971}
F.~{Barbaresco}.
\newblock Coding statistical characterization of radar signal fluctuation for
  {Lie} group machine learning.
\newblock In {\em 2019 International Radar Conference (RADAR)}, pages 1--6,
  2019.

\bibitem{BARBARESCO201054}
F.~Barbaresco and Uwe Meier.
\newblock Radar monitoring of a wake vortex: Electromagnetic reflection of wake
  turbulence in clear air.
\newblock {\em Comptes Rendus Physique}, 11(1):54--67, 2010.

\bibitem{4154721}
O.~{Besson}, J.~{Tourneret}, and S.~{Bidon}.
\newblock Knowledge-aided {B}ayesian detection in heterogeneous environments.
\newblock {\em IEEE Signal Processing Letters}, 14(5):355--358, 2007.

\bibitem{4359546}
S.~{Bidon}, O.~{Besson}, and J.~{Tourneret}.
\newblock A {B}ayesian approach to adaptive detection in nonhomogeneous
  environments.
\newblock {\em IEEE Transactions on Signal Processing}, 56(1):205--217, 2008.

\bibitem{8707051}
N.~{Bouhlel} and A.~{Dziri}.
\newblock {K}ullback--{L}eibler divergence between multivariate generalized
  {G}aussian distributions.
\newblock {\em IEEE Signal Processing Letters}, 26(7):1021--1025, 2019.

\bibitem{bridson2013metric}
M.~R. Bridson and A.~H\"afliger.
\newblock {\em Metric Spaces of Non-Positive Curvature}, volume 319.
\newblock Springer Science \& Business Media, 2013.

\bibitem{CCMV2013}
M.~Charfi, Z.~Chebbi, M.~Moakher, and B.~C. Vemuri.
\newblock Bhattacharyya median of symmetric positive-definite matrices and
  application to the denoising of diffusion-tensor fields.
\newblock In {\em Proc IEEE Int Symp Biomed Imaging}, pages 1227--1230, 2013.

\bibitem{6252067}
X.~{Chen}, Z.~J. {Wang}, and M.~J. {McKeown}.
\newblock Shrinkage-to-tapering estimation of large covariance matrices.
\newblock {\em IEEE Transactions on Signal Processing}, 60(11):5640--5656,
  2012.

\bibitem{5484583}
Y.~{Chen}, A.~{Wiesel}, Y.~C. {Eldar}, and A.~O. {Hero}.
\newblock Shrinkage algorithms for {MMSE} covariance estimation.
\newblock {\em IEEE Transactions on Signal Processing}, 58(10):5016--5029,
  2010.

\bibitem{6378374}
A.~{Cherian}, S.~{Sra}, A.~{Banerjee}, and N.~{Papanikolopoulos}.
\newblock {J}ensen--{B}regman {L}og{D}et divergence with application to
  efficient similarity search for covariance matrices.
\newblock {\em IEEE Transactions on Pattern Analysis and Machine Intelligence},
  35(9):2161--2174, 2013.

\bibitem{7605536}
D.~{Ciuonzo}, A.~{De Maio}, and D.~{Orlando}.
\newblock On the statistical invariance for adaptive radar detection in
  partially homogeneous disturbance plus structured interference.
\newblock {\em IEEE Transactions on Signal Processing}, 65(5):1222--1234, 2017.

\bibitem{511809}
E.~{Conte}, M.~{Lops}, and G.~{Ricci}.
\newblock Adaptive matched filter detection in spherically invariant noise.
\newblock {\em IEEE Signal Processing Letters}, 3(8):248--250, 1996.

\bibitem{5417154}
A.~{De Maio}, A.~{Farina}, and G.~{Foglia}.
\newblock Knowledge-aided {B}ayesian radar detectors \& their application to
  live data.
\newblock {\em IEEE Transactions on Aerospace and Electronic Systems},
  46(1):170--183, Jan 2010.

\bibitem{7426844}
A.~{De Maio}, D.~{Orlando}, C.~{Hao}, and G.~{Foglia}.
\newblock Adaptive detection of point-like targets in spectrally symmetric
  interference.
\newblock {\em IEEE Transactions on Signal Processing}, 64(12):3207--3220,
  2016.

\bibitem{7842633}
A.~{Decurninge} and F.~{Barbaresco}.
\newblock Robust {B}urg estimation of radar scatter matrix for autoregressive
  structured {SIRV} based on {F}r{\'e}chet medians.
\newblock {\em IET Radar, Sonar \& Navigation}, 11(1):78--89, 2017.

\bibitem{7478057}
J.~{Dun\'ik}, O.~{Straka}, and M.~{\v{S}imandl}.
\newblock On autocovariance least-squares method for noise covariance matrices
  estimation.
\newblock {\em IEEE Transactions on Automatic Control}, 62(2):967--972, 2017.

\bibitem{Gavaskar2019Fast}
R.~G. Gavaskar and K.~N. Chaudhury.
\newblock Fast adaptive bilateral filtering.
\newblock {\em IEEE Transactions on Image Processing}, 28(2):779--790, 2019.

\bibitem{GLIGA2019381}
L.~I. Gliga, H.~Chafouk, D.~Popescu, and C.~Lupu.
\newblock A method to estimate the process noise covariance for a certain class
  of nonlinear systems.
\newblock {\em Mechanical Systems and Signal Processing}, 131:381--393, 2019.

\bibitem{Hig2008}
N.~J. Higham.
\newblock {\em Functions of Matrices: Theory and Computation}.
\newblock SIAM, Philadelphia, 2008.

\bibitem{e20040219}
X.~Hua, Y.~Cheng, H.~Wang, and Y.~Qin.
\newblock Robust covariance estimators based on information divergences and
  {R}iemannian manifold.
\newblock {\em Entropy}, 20(4), 2018.

\bibitem{HUA2018232}
X.~Hua, Y.~Cheng, H.~Wang, Y.~Qin, and D.~Chen.
\newblock {Geometric target detection based on total Bregman divergence}.
\newblock {\em Digital Signal Processing}, 75:232--241, 2018.

\bibitem{8000811}
X.~{Hua}, Y.~{Cheng}, H.~{Wang}, Y.~{Qin}, and Y.~{Li}.
\newblock Geometric means and medians with applications to target detection.
\newblock {\em IET Signal Processing}, 11(6):711--720, 2017.

\bibitem{HUA2017106}
X.~Hua, Y.~Cheng, H.~Wang, Y.~Qin, Y.~Li, and W.~Zhang.
\newblock {Matrix CFAR detectors based on symmetrized Kullback--Leibler and
  total Kullback--Leibler divergences}.
\newblock {\em Digital Signal Processing}, 69:106--116, 2017.

\bibitem{e20040256}
X.~Hua, H.~Fan, Y.~Cheng, H.~Wang, and Y.~Qin.
\newblock Information geometry for radar target detection with total
  {J}ensen--{B}regman divergence.
\newblock {\em Entropy}, 20(4), 2018.

\bibitem{huaetal2020}
X.~Hua, Y.~Ono, L.~Peng, Y.~Cheng, and H.~Wang.
\newblock Target detection within nonhomogeneous clutter via total {B}regman
  divergence-based matrix information geometry detectors.
\newblock 2020.
\newblock arXiv:2012.13861.

\bibitem{7073532}
B.~{Kang}, V.~{Monga}, and M.~{Rangaswamy}.
\newblock {Computationally efficient Toeplitz approximation of structured
  covariance under a rank constraint}.
\newblock {\em IEEE Transactions on Aerospace and Electronic Systems},
  51(1):775--785, 2015.

\bibitem{4104190}
E.~J. {Kelly}.
\newblock An adaptive detection algorithm.
\newblock {\em IEEE Transactions on Aerospace and Electronic Systems},
  AES-22(2):115--127, 1986.

\bibitem{782198}
S.~{Kraut} and L.~L. {Scharf}.
\newblock The {CFAR} adaptive subspace detector is a scale-invariant {GLRT}.
\newblock {\em IEEE Transactions on Signal Processing}, 47(9):2538--2541, 1999.

\bibitem{4721049}
J.~{Lapuyade-Lahorgue} and F.~{Barbaresco}.
\newblock {Radar detection using Siegel distance between autoregressive
  processes, application to HF and X-band radar}.
\newblock In {\em 2008 IEEE Radar Conference}, pages 1--6, 2008.

\bibitem{8747386}
J.~{Li}, A.~{Aubry}, A.~{De Maio}, and J.~{Zhou}.
\newblock An {EL} approach for similarity parameter selection in {KA}
  covariance matrix estimation.
\newblock {\em IEEE Signal Processing Letters}, 26(8):1217--1221, 2019.

\bibitem{Liu2013}
Z.~Liu and F.~Barbaresco.
\newblock Doppler information geometry for wake turbulence monitoring.
\newblock In Frank Nielsen and Rajendra Bhatia, editors, {\em Matrix
  Information Geometry}, pages 277--290, Berlin, Heidelberg, 2013. Springer.

\bibitem{8920217}
G.~{Luo}, J.~{Wei}, W.~{Hu}, and S.~J. {Maybank}.
\newblock {Tangent Fisher Vector on Matrix Manifolds for Action Recognition}.
\newblock {\em IEEE Transactions on Image Processing}, 29:3052--3064, 2020.

\bibitem{4267625}
A.~D. {Maio}, A.~{Farina}, and G.~{Foglia}.
\newblock Design and experimental validation of knowledge-based constant false
  alarm rate detectors.
\newblock {\em IET Radar, Sonar \& Navigation}, 1(4):308--316, Aug 2007.

\bibitem{5210021}
A.~D. {Maio}, S.~D. {Nicola}, L.~{Landi}, and A.~{Farina}.
\newblock Knowledge-aided covariance matrix estimation: a {MAXDET} approach.
\newblock {\em IET Radar, Sonar \& Navigation}, 3(4):341--356, 2009.

\bibitem{Moa2005}
M.~Moakher.
\newblock A differential geometric approach to the geometric mean of symmetric
  positive-definite matrices.
\newblock {\em SIAM Journal on Matrix Analysis and Applications},
  26(3):735--747, 2005.

\bibitem{8517121}
P.~{Nair} and K.~N. {Chaudhury}.
\newblock Fast high-dimensional bilateral and nonlocal means filtering.
\newblock {\em IEEE Transactions on Image Processing}, 28(3):1470--1481, 2019.

\bibitem{8895800}
S.~{Park}, A.~{Ali}, N.~{Gonz\'alez-Prelcic}, and R.~W. {Heath}.
\newblock Spatial channel covariance estimation for hybrid architectures based
  on tensor decompositions.
\newblock {\em IEEE Transactions on Wireless Communications}, 19(2):1084--1097,
  2020.

\bibitem{135446}
F.~C. {Robey}, D.~R. {Fuhrmann}, E.~J. {Kelly}, and R.~{Nitzberg}.
\newblock A {CFAR} adaptive matched filter detector.
\newblock {\em IEEE Transactions on Aerospace and Electronic Systems},
  28(1):208--216, 1992.

\bibitem{sra2016positive}
S.~Sra.
\newblock {Positive definite matrices and the S-divergence}.
\newblock {\em Proceedings of the American Mathematical Society},
  144(7):2787--2797, 2016.

\bibitem{6853408}
P.~{Wang}, Z.~{Wang}, H.~{Li}, and B.~{Himed}.
\newblock Knowledge-aided parametric adaptive matched filter with automatic
  combining for covariance estimation.
\newblock {\em IEEE Transactions on Signal Processing}, 62(18):4713--4722,
  2014.

\bibitem{8060999}
K.~M. {Wong}, J.~{Zhang}, J.~{Liang}, and H.~{Jiang}.
\newblock {Mean and Median of PSD Matrices on a Riemannian Manifold:
  Application to Detection of Narrow-Band Sonar Signals}.
\newblock {\em IEEE Transactions on Signal Processing}, 65(24):6536--6550,
  2017.

\bibitem{8624413}
O.~{Yair}, M.~{Ben-Chen}, and R.~{Talmon}.
\newblock {Parallel Transport on the Cone Manifold of SPD Matrices for Domain
  Adaptation}.
\newblock {\em IEEE Transactions on Signal Processing}, 67(7):1797--1811, 2019.

\end{thebibliography}

\begin{appendix}
\section{Anisotropy index of the JBLD divergence: Proof of Proposition \ref{prop:aiJ} }
\label{appen:A}

Denote $f(\varepsilon)$ as the function to be minimized, namely
\begin{equation*}
f(\varepsilon) = \operatorname{ln}\operatorname{det}\bigg(\frac{\bm{R} + \varepsilon\bm{I}}{2}\bigg) - \frac{1}{2}\operatorname{ln}\operatorname{det}(\varepsilon\bm{R}),
\end{equation*}
whose derivative can be computed directly and reads
\begin{equation*}
\begin{aligned}
f'(\varepsilon)= \operatorname{tr}\left( \left(\bm{R}+\varepsilon\bm{I} \right)^{-1} -\frac{1}{2\varepsilon}\bm{I}\right) .
\end{aligned}
\end{equation*}
Let $\{\lambda_1,\lambda_2,\ldots,\lambda_n\}$ be the eigenvalues of the matrix $\bm{R}$, and we have
\begin{equation*}
\operatorname{tr}\Big( ( \bm{R}+\varepsilon\bm{I} )^{-1} \Big) = \sum_{i=1}^n \frac{1}{\lambda_i + \varepsilon}.
\end{equation*}
Then, the equation $f'(\varepsilon)=0$ becomes
\begin{equation*}
2\sum_{i=1}^n \frac{1}{\lambda_i + \varepsilon} = \frac{n}{\varepsilon}.
\end{equation*}
It can be rearranged as an $n$-th order polynomial for the unknown variable $\varepsilon$, which has one unique positive solution for all $n\geq 1$, denoted by $\varepsilon^*$. The unique existence is verified in the next Lemma \ref{lem:npe}.

\begin{lem}
\label{lem:npe}
The rational equation
\begin{equation*}
2\sum_{i=1}^n \frac{1}{\lambda_i + \varepsilon} = \frac{n}{\varepsilon}
\end{equation*}
has only one positive solution $\varepsilon^*$ for all $n\geq 1$.
\end{lem}

\begin{proof}
First of all, we rearrange the equation into the vanishment of a polynomial $g_n(\varepsilon)=0$ ($n\geq 1$)  where
\begin{equation*}
\begin{aligned}
g_n(\varepsilon)&=2\varepsilon\left(\sum_{i=1}^n(\lambda_1+\varepsilon)(\lambda_2+\varepsilon)\cdots (\widehat{\lambda_i+\varepsilon})\cdots (\lambda_n+\varepsilon)\right) -n(\lambda_1+\varepsilon)(\lambda_2+\varepsilon)\cdots (\lambda_n+\varepsilon),
\end{aligned}
\end{equation*}
satisfying
\begin{equation*}
2\sum_{i=1}^n \frac{1}{\lambda_i + \varepsilon} - \frac{n}{\varepsilon}=\frac{g_n(\varepsilon)}{\varepsilon \prod_{i=1}^n(\lambda_i + \varepsilon) }.
\end{equation*}
Here $\widehat{\lambda_i+\varepsilon}$ means this term does not appear in the calculation.
As all eigenvalues are positive, searching for positive solutions of the rational equation is equivalent to finding positive solutions of $g_n(\varepsilon)=0$.

By  mathematical induction,   we have
\begin{equation*}
\begin{aligned}
g_n(\varepsilon)&=n\varepsilon^n+(n-2)\varepsilon^{n-1}\sum_{i=1}^n\lambda_i+(n-4)\varepsilon^{n-2}\sum_{ i_1<i_2}\lambda_{i_1}\lambda_{i_2}\\
&~~~~+\cdots+(n-2k)\varepsilon^{n-k}\sum_{i_1<i_2<\cdots <i_k}\lambda_{i_1}\lambda_{i_2}\cdots \lambda_{i_k} +\cdots+(-n)\prod_{i=1}^n\lambda_1\lambda_2\cdots\lambda_n\\
&=\sum_{k=0}^n(n-2k)\varepsilon^{n-k}\sum_{i_1<i_2<\cdots <i_k}\lambda_{i_1}\lambda_{i_2}\cdots \lambda_{i_k}.
\end{aligned}
\end{equation*}

We will consider when $n$ is even only; the proof for an odd $n$ is analogous. Note that there is no $\varepsilon^{n/2}$ term in $g_n(\varepsilon)$, and all coefficients for the terms $\varepsilon^{n-k}$, $k=0,1,2,\ldots,n/2-1$, are positive  while all coefficients for the terms $\varepsilon^{n-k}$, $k=n/2+1,n/2+2,\ldots,n$, are negative.  Starting from $(n/2)$-th order derivative of $g_n(\varepsilon)$, we have
\begin{equation*}
g_n^{(n/2)}(\varepsilon)>0,\text{ for all } \varepsilon>0.
\end{equation*}
Furthermore, since
\begin{equation*}
g_n^{(n/2-1)}(0)<0 \text{ and } \lim\limits_{\varepsilon\rightarrow +\infty}g_n^{(n/2-1)}(\varepsilon)=+\infty,
\end{equation*}
 we know  $g_n^{(n/2-1)}(\varepsilon)$ passes the $\varepsilon$-axis once and only once for $\varepsilon>0$.
Again, since
\begin{equation*}
g_n^{(n/2-2)}(0)<0 \text{ and } \lim\limits_{\varepsilon\rightarrow +\infty}g_n^{(n/2-2)}(\varepsilon)=+\infty,
\end{equation*}
the function $g_n^{(n/2-2)}(\varepsilon)$ decreases from a negative value and then increases to positive infinity. Consequently, $g_n^{(n/2-2)}(\varepsilon)$ also passes the $\varepsilon$-axis once and only once for $\varepsilon>0$.

Continuing this procedure of analysis, we conclude that the function $g_n(\varepsilon)$   passes the $\varepsilon$-axis only once for $\varepsilon>0$. Namely, $g_n(\varepsilon)=0$ has only one positive solution.

\end{proof}

\end{appendix}

\end{document}